\documentclass[envcountsame]{llncs}
\spnewtheorem{definition}{Definition}{\bfseries}{\rmfamily}

\usepackage{amssymb,latexsym}

\usepackage{amsthm}
\usepackage{adjustbox}
\usepackage{algorithm}
\usepackage{enumerate}
\usepackage{appendix}
\usepackage{amsmath}
\usepackage[numbers]{natbib}
\makeatletter
\renewcommand\bibsection%
{
  \section*{\refname
  \@mkboth{\MakeUppercase{\refname}}{\MakeUppercase{\refname}}}
}
\makeatother

\usepackage{mathtools}

\DeclarePairedDelimiter{\floor}{\lfloor}{\rfloor}

\usepackage{etoolbox}
\makeatletter
    \pretocmd{\NAT@citexnum}{\@ifnum{\NAT@ctype>\z@}{\let\NAT@hyper@\relax}{}}{}{}
\makeatother

\newcommand{\CCite}[1]{\citeauthor{#1}~\cite{#1}}

\usepackage{xspace} 

\usepackage{pgf,tikz}
\usepackage{wrapfig}  
\usepackage{cutwin}  

\usepackage{caption}
\usepackage{subcaption}
\usepackage{hyperref}
\hypersetup{%
   breaklinks,%
   colorlinks=true,%
   linkcolor=[rgb]{0.45,0.0,0.0},%
   urlcolor=[rgb]{0.05,0.390,0.0},%
   citecolor=[rgb]{0,0,0.45}
}%

\usepackage[nameinlink]{cleveref} 
\usepackage[noend]{algorithmic}
\Crefname{ALC@unique}{Line}{Lines} 

\crefname{ineq}{inequality}{inequalities}
\creflabelformat{ineq}{#2{\upshape(#1)}#3} 
\crefname{claim}{claim}{claims}           
\crefname{definition}{definition}{definitions}  
\crefname{obs}{Observation}{Observations}
\crefname{lemma}{Lemma}{Lemmas}
\crefname{cor}{Corollary}{Corollaries}

\usetikzlibrary{arrows,shapes,calc,intersections,through,backgrounds,patterns}

\theoremstyle{plain}
\newtheorem{obs}[definition]{Observation}
\newtheorem{cor}[theorem]{Corollary}
\usepackage{thmtools, thm-restate} 

\newcommand{\Real}{\ensuremath{\mathbb{R}}\xspace}%
\newcommand{\Int}{\ensuremath{\mathbb{Z}}\xspace}%
\newcommand{\Intp}{\ensuremath{\mathbb{Z}_{+} \cup \{0\}}\xspace}%
\newcommand{\NP}{\ensuremath{\mathbb{NP}}\xspace}%
\newcommand{\APX}{\ensuremath{\mathbb{APX}}\xspace}%
\newcommand{\PNP}{\ensuremath{\mathbb{P} = \mathbb{NP}}\xspace}%

\newcommand{\lay}{\ensuremath{\lambda}\xspace}%
\newcommand{\layp}{\ensuremath{\lambda'}\xspace}%
\newcommand{\lays}{\ensuremath{\lambda^*}\xspace}%
\newcommand{\il}{\ensuremath{\lambda^{\text{in}}}\xspace}%
\newcommand{\ilx}{\ensuremath{\lambda^{\text{in}}_x}\xspace}%
\newcommand{\ily}{\ensuremath{\lambda^{\text{in}}_y}\xspace}%

\newcommand{\lxr}{\ensuremath{\lambda_x(r)}\xspace}%
\newcommand{\lyr}{\ensuremath{\lambda_y(r)}\xspace}%
\newcommand{\lxrp}{\ensuremath{\lambda_x(r')}\xspace}%
\newcommand{\lyrp}{\ensuremath{\lambda_y(r')}\xspace}%

\newcommand{\pxr}{\ensuremath{\lambda'_x(r)}\xspace}%
\newcommand{\pyr}{\ensuremath{\lambda'_y(r)}\xspace}%
\newcommand{\pxrp}{\ensuremath{\lambda'_x(r')}\xspace}%
\newcommand{\pyrp}{\ensuremath{\lambda'_y(r')}\xspace}%
\newcommand{\lin}{\ensuremath{\tau}\xspace}%
\newcommand{\Lin}{\ensuremath{L}\xspace}%
\newcommand{\seg}{\ensuremath{\phi}\xspace}%
\newcommand{\Seg}{\ensuremath{\Phi}\xspace}%
\newcommand{\con}{\ensuremath{\xi}}%
\newcommand{\apconst}{\ensuremath{\frac{1}{55}}\xspace}%

\mainmatter
\begin{document}
\title{On the Approximability of Orthogonal Order Preserving Layout Adjustment}
\titlerunning{An $O(1)$ Approximation for Orthogonal Order Preserving Layout
Adjustment} 

\author{Sayan Bandyapadhyay
        \and Santanu Bhowmick
        \and Kasturi Varadarajan}
\institute{
  Department of Computer Science\\
  University of Iowa, Iowa City, USA
  }
  
\authorrunning{S.\,Bandyapadhyay, S.\,Bhowmick and K.\,Varadarajan} 

\maketitle

\vspace{-0.4cm}
\begin{abstract}
Given an initial placement of a set of rectangles in the plane, we consider the
problem of finding a disjoint placement of the rectangles that minimizes the
area of the bounding box and preserves the orthogonal order i.e.\ maintains the
sorted ordering of the rectangle centers along both $x$-axis and $y$-axis with respect
to the initial placement. This problem is known as \textit{Layout Adjustment for
Disjoint Rectangles} (LADR). It was known that LADR is \NP-hard, but only
heuristics were known for it. We show that a certain decision version of LADR is
\APX-hard, and give a constant factor approximation for LADR.
\end{abstract} 

\section{Introduction}
\label{sec:intro}
Graphs are often used to visualize relationships between entities in diverse
fields such as software engineering (e.g.\ UML diagrams), VLSI (circuit
schematics) and biology (e.g.\ biochemical pathways)~\cite{herman00}. For many
such applications, treating graph nodes as points is insufficient, since each
node may have a corresponding label explaining its significance. The presence of
labels may lead to node overlapping. For the typical user, an uncluttered layout
is more important than
the amount of information presented~\cite{storey96}. For complex graphs, it is
tedious to create meaningful layouts by hand, which has led to algorithms for
layout generation.

Layout generation algorithms typically take a combinatorial description of a
graph, and return a corresponding layout. Nodes are usually represented by
boxes, and edges by lines connecting the boxes. For simplicity, the
edges of the graph are ignored while creating the modified layout. In some interactive
systems, modifications to the graph may happen in multiple stages. The layout
must be adjusted after each alteration (if new nodes added overlap
existing nodes), such that the display area is minimized.  If we use layout
creation algorithms after each iteration, we may get a layout that is completely
different from the previous layout, which may destroy the `mental map' of the
user who is interacting with the system. Thus, we need an additional constraint
in the form of maintaining some property of the layout, which would be
equivalent to preserving the mental map.~\CCite{eades1991preserving} defined
\textit{orthogonal ordering} as one of the key properties that should be
maintained in an adjusted layout to preserve the user's mental map. Two layouts
of a graph have the same orthogonal ordering if the horizontal and vertical
ordering of the nodes are identical in both layouts.

We now state the problem studied in this paper, which involves laying out
rectangles that represent the nodes in the graph being adjusted.
We are given a set of rectangles $R$ (each $r_i \in R$ is defined by an ordered
pair, $r_i = (w_i, h_i)$, denoting its width and height respectively) and an
initial layout \il. A layout consists of an assignment $\lay: R \rightarrow
\Real^2$ of coordinates to the centers of rectangles in $R$. The goal is
to find a layout in which no two rectangles intersect and orthogonal ordering of
the rectangle centers  w.r.t \il is maintained, while minimizing the area of the
bounding box of the layout. We refer to this problem as \textit{Layout
Adjustment for Disjoint Rectangles} (LADR). Note that $R$ is really a set of
rectangle dimensions, and not a set of rectangles. Nevertheless, we will refer
to $R$ as a set of rectangles. See \Cref{sec:prelim} for a more leisurely
problem statement.

\vspace{-0.2cm}
\subsection{Previous Work}
\vspace{-0.1cm}
The concept of a mental map was introduced in~\cite{eades1991preserving}, along
with three quantitative models representing it - orthogonal ordering, proximity
relations and topology. A framework for analyzing the various models of a mental
map was presented in~\cite{Bridgeman98}, which determined that orthogonal
ordering constraint was the best metric for comparing different drawings of the
same graph. A user study designed to evaluate human perceptions of similarity
amongst two sets of drawings was given in~\cite{Bridgeman02}, in which
orthogonal ordering constraints received the highest rankings.

There has been a lot of work done using the concept of preserving mental maps.
LADR was first introduced in~\cite{misue95},
in which the authors described the \textit{Force-Scan (FS)} algorithm. FS scans for
overlapping nodes in both horizontal and vertical directions, and separates two
intersecting nodes by ``forcing'' them apart along the line connecting the centers of
the two nodes, while ensuring that the nodes being forced apart do not intersect
any additional nodes in the layout. In~\cite{hayashi98}, a modification of FS
was presented (FS'), which resulted in a more compact layout than FS. Another
version of FS algorithm, called the \textit{Force-Transfer (FT)} algorithm, was
given in~\cite{Huang2007}. For any two overlapping nodes, denote the vertical
distance to be moved to remove the overlap as $d_v$, and let the horizontal
distance for removing overlap be $d_h$. FT moves the overlapping node
horizontally if $d_h < d_v$, else vertically, and experimentally, it has been
shown that FT gives a layout of smaller area than FS and FS'.

FS, FS' and FT belong to the family of force based layout algorithms. Spring based
algorithms treat edges as springs obeying Hooke's Law, and the nodes are pushed
apart or pulled in iteratively to balance the forces till an equilibrium is
reached. A spring based algorithm \textit{ODNLS}, which adjusts the
attractive/repulsive force between two nodes dynamically, is proposed 
in~\cite{LiEN05}, which preserves the orthogonal ordering of the input layout
and typically returns a smaller overlap-free layout than the force-based family of
algorithms.

It is worth noting that none of the algorithms mentioned above give a provable
worst-case guarantee on the quality of the output.

The hardness of preserving orthogonal constraints w.r.t various optimization
metrics has also been well-studied.~\CCite{BrandesP13} showed that it is
\NP-hard to determine if there exists an orthogonal-order preserving rectilinear
drawing of a simple path, and extend the result for determination of uniform
edge-length drawings of simple paths with same constraints. LADR was shown to be
\NP-hard by~\CCite{hayashi98}, using a reduction from $3SAT$.

\subsection{Related Work}
Algorithms for label placement and packing that do not account for orthogonal
ordering have been extensively studied. The placement of labels corresponding to points on a map is a natural
problem that arises in geographic information systems (GIS)~\cite{MapLabel}. In
particular, placing labels on maps such that the
label boundary coincides with the point feature has been a well-studied problem.
A common objective in such label-placement problems is to maximize the number of
features labelled, such that the labels are pairwise disjoint. We refer
to~\cite{AgarwalKS98,KreveldSW99} as examples of this line of work.

Packing rectangles without orthogonality constraints has also been
well-studied. One such problem is the strip packing problem, in which we want
to pack a set of rectangles into a strip of given width while minimizing the
height of the packing. It is known that the strip-packing problem is strongly
\NP-hard~\cite{LodiMM02}. It can be easily seen that if the constraint for
orthogonal order preservation is removed, then LADR can be reduced to
multiple instances of strip packing problem. There has been extensive work done
on strip packing~\cite{Steinberg97,Schiermeyer94,HarrenS09}, with the current
best algorithm being a $5/3 + \varepsilon$-approximation by~\CCite{HarrenJPS14}. 

Another related packing problem is the two-dimensional geometric knapsack
problem, defined as follows. The input consists of a set of weighted rectangles
and a rectangular knapsack, and the goal is to find a subset of rectangles of
maximum weight that can be placed in the knapsack such that no two rectangles
have an overlap. The 2D-knapsack problem is known to be strongly \NP-hard even
when the input consists of a set of unweighted squares~\cite{LodiMM02}.
Recently, \CCite{AdamaszekW15} gave a quasi-polynomial time $(1+\varepsilon)$
approximation scheme for this problem, with the assumption that the input
consists of quasi-polynomially bounded integers. 

\vspace{-0.2cm}
\subsection{Our results}
\vspace{-0.1cm}
We point out an intimate connection between LADR and the problem of hitting
segments using a minimum number of horizontal and vertical lines. In particular,
the segments to be hit are the ones connecting each pair of rectangle centers in
the input layout. The connection to the hitting set is described in
\Cref{sec:reduction}. To our knowledge, this connection to hitting sets has not
been observed in the literature.  
We exploit the connection to hitting set to prove hardness results for
LADR in \Cref{sec:APXH} that complement the \NP-completeness result
in~\cite{hayashi98}. We show that it is \APX-hard to find a layout that
minimizes the perimeter of the bounding box. We also show that if there is an
approximate decision procedure that determines whether there is a layout that
fits within a bounding box of specified dimensions, then \PNP. These hardness
results hold even when the input rectangles are unit squares. The results for
LADR follow from a hardness of approximation result that we show for a hitting
set problem. The starting point of the latter is the result of \CCite{HassinM91}
who show that it is \NP-hard to determine if there is a set of $k$ axis-parallel
lines that hit a set of horizontal segments of unit length. The added difficulty
that we need to overcome is that in our case, the set of segments that need to
be hit cannot be arbitrarily constructed. Rather, the set consists of all segments
induced by a set of arbitrarily constructed points.

It is possible to exploit this connection to hitting sets and use known
algorithms for hitting sets (e.g.~\cite{GaurIK02}) to devise an $O(1)$
approximation algorithm for LADR. Instead, we describe (in \Cref{sec:approx}) a
direct polynomial time algorithm for LADR that
achieves a $4(1 + o(1))$ approximation. This is the first polynomial time
algorithm for LADR with a provable approximation guarantee. The algorithm
involves solving a linear-programming relaxation of LADR followed by a simple
rounding. 

\section{Preliminaries}
\label{sec:prelim}
We define a layout \lay of a set of rectangles $R$ as an assignment of
coordinates to the center of each rectangle $r \in R$ i.e.\ $\lay: R \rightarrow
\Real^2$. Our input for LADR consists of a set of rectangles $R$, and an
initial layout \il. We will assume that \il is injective, i.e.\ no two rectangle
centers coincide in the input layout.  A rectangle $r$ is defined by its
horizontal width $w_r$ and vertical height $h_r$, both of which are assumed to
be integral. 
It is given that all rectangles are axis-parallel in \il, and rotation of
rectangles is not allowed in any adjusted layout.  

The coordinates of center of $r$ in layout \lay is denoted by $\lay(r) = (x_r,
y_r)$. For brevity, we denote the $x$-coordinate of $\lay(r)$ by \lxr, and the
corresponding $y$-coordinate by \lyr. The set of points $\{\lay(r): r \in R \}$
is denoted by $\lay(R)$. 

A pair of rectangles $r, r' \in R$ is said to intersect in a layout
\lay if and only if 
\begin{align}
    \label[ineq]{ineq:disjoint}
    |\lxr - \lxrp| < \frac{w_r + w_{r'}}{2} \qquad \text{and} \qquad 
     |\lyr - \lyrp| < \frac{h_r + h_{r'}}{2}.
\end{align}

A layout \lay is termed as a \textit{disjoint} layout if no two rectangles in
$R$ intersect with each other. Let $W_l(\lay)$ and $W_r(\lay)$ denote the
$x$-coordinates of the left and right sides of the smallest axis-parallel
rectangle bounding the rectangles of $R$ placed by \lay, respectively. We then
define the width of the layout, $W(\lay) = W_r(\lay) - W_l(\lay)$. Similarly,
let $H_t(\lay)$ and $H_b(\lay)$ define the $y$-coordinates of the top and bottom
of the bounding rectangle, and the height of the layout is defined as $H(\lay) =
H_t(\lay) - H_b(\lay)$. The area of \lay is thus defined as $A(\lay) =
H(\lay) \times W(\lay)$. The perimeter of \lay is $2 (H(\lay) + W(\lay))$. 

Let \lay and \layp be two layouts of $R$.  Then, \lay and
\layp are defined to have the same \textit{orthogonal ordering} if for any
two rectangles $r, r' \in R$,

\vspace{-0.5cm}
\begin{align}
  \label{eq:xineq}
  \qquad \lxr < \lxrp \  \iff \pxr < \pxrp \\
  \label{eq:yineq}
  \qquad \lyr < \lyrp \  \iff \pyr < \pyrp \\ 
  \label{eq:xrig}
  \qquad \lxr = \lxrp \  \iff \pxr = \pxrp \\
  \label{eq:yrig}
  \qquad \lyr = \lyrp \  \iff \pyr = \pyrp
\end{align}

For any $R$ and corresponding \il, the minimal area of a layout 
is defined as:
$ A^{\min} = \inf \{ A(\lay) : \lay\ \text{is a disjoint layout},
\lay\ \text{has same orthogonal ordering as \il} \}$
It should be noted that it may not be possible to attain a disjoint
orthogonality preserving layout whose area is the same as $A^{\min}$ -
we can only aim to get a layout whose area is arbitrarily close to $A^{\min}$.

We introduce the concept of rigidities in a layout. A set of rectangles $R'
\subseteq R$ forms a $x$-rigidity in a layout \lay if $\exists \alpha$ such that
$R' = \{ r \in R \mid \lay_x(r) = \alpha \}$. We define a $y$-rigidity in a layout 
similarly in terms of the $y$-coordinates of the rectangles in that layout. We
observe that any rectangle $r$ belongs to a unique $x$-rigidity (and a unique
$y$-rigidity), which may consist of merely itself and no other rectangle. We
order the $x$-rigidities in a layout \lay in increasing order of
$x$-coordinates, and for any rectangle $r \in R$, we define its $x$-rank to be
$i$ if $r$ belongs to the $i$-th $x$-rigidity in this ordering. It is obvious
that $x$-rank of any rectangle is an integer between 1 and $|R|$. We similarly
define the $y$-rank of each rectangle in terms of its $y$-rigidities. Unless
otherwise stated, we refer to the rigidities and ranks of the initial layout \il
whenever these terms are used in the paper.

Let $\seg(p, p')$ be the segment whose endpoints are points $p, p' \in P$. Then
the set of segments induced by a set of points $P$ is defined as $\Seg(P) =
\{\seg(p, p') : p,p' \in P,\ p \neq p' \}$, denoted by \Seg when $P$ is clear
from the context.

We also consider a simpler version of LADR where the set of rectangles
$R$ consists of unit squares. We call this version as the \textit{Layout
Adjustment for Disjoint Squares} problem, and refer to it as LADS for brevity.

\vspace{-0.3cm}
\section{Reduction of LADS to Hitting Set}
\label{sec:reduction}
\vspace{-0.2cm}
We formally define a unit grid as follows. Let $f: \Real^2 \rightarrow \Int^2$
be the function $f(x,y) = \left ( \floor{x}, \floor{y} \right)$. The function
$f$ induces a partition of $\Real^2$ into grid cells - grid cell $(i,j)$ is the
set $\{p \in \Real^2 \mid f(p) = (i,j) \}$. We call this partition a unit grid
on $\Real^2$. The `grid lines' are the vertical lines $x = \alpha$ and $y =
\alpha$ for integer $\alpha$.

Let $S$ be the set of unit squares provided as input to LADS, having initial
layout \il. Consider a disjoint, orthogonal order preserving layout \lay for
$S$ . Let \Lin be the subset consisting of those grid lines that intersect the
minimum bounding box of $\lay(S)$. Let \seg be the line segment connecting the
points $\lay(s)$ and $\lay(s')$, for some $s, s' \in S$. Since the layout \lay
is disjoint, $\lay(s)$ and $\lay(s')$ lie in different grid cells. Thus, there
exists at least one line $\lin \in \Lin$ that intersects \seg. Motivated by
this, we define a hitting set problem as follows.

We say a line \lin hits a line segment \seg if \lin intersects the relative
interior of \seg but not either end point of \seg.
Thus, if \seg is a horizontal line segment (which would happen if
$s, s'$ belongs to a $y$-rigidity), then \seg cannot be hit by a
horizontal line $\lin \in L$. We thus define the Uniform Hitting Set (UHS)
problem as follows: 

\begin{definition}[Uniform Hitting Set - Decision Problem]
   \label{def:UHS}
   Given a set of segments \Seg induced by a point set $P$ and a non-negative
   integer $k$, is there a set of axis-parallel lines \Lin that hit all segments
   in \Seg, such that $|L| \le k$?
\end{definition}

Since the area of the minimum bounding box for $\lay(S)$ is roughly the product
of the number of horizontal grid lines intersecting it and the number of
vertical grid lines intersecting it, we also need the following variant.

\begin{definition}[Constrained Uniform Hitting Set - Decision Problem]
    \label{def:CUHS}
    Given a set of line segments, \Seg, induced by a set of points $P$, and
    non negative integers $r,c$, is it possible to hit all segments in \Seg with a
    set of lines \Lin containing at most $r$ horizontal lines and $c$ vertical
    lines ?
\end{definition}

The term `uniform' in the problem name refers to the fact that each segment in
\Seg needs to be hit only once by a horizontal or vertical line. We denote the
problem thus defined as CUHS, and proceed to show its equivalence with
a constrained version of the layout adjustment problem. 

\begin{definition}[Constrained LADS - Decision Problem]
   \label{def:CLADS}
   Given $n$ unit squares $S$, initial layout \il, positive integers $w, h$ and
   a constant $0 < \varepsilon < 1$, is there a layout \layp having height
   $H(\layp) \leq h + \varepsilon$ and width $W(\layp) \leq w + \varepsilon$,
   satisfying the following conditions?
   \begin{enumerate}
       \item \layp is a disjoint layout.
       \item \il and \layp have the same orthogonal order.
   \end{enumerate}
\end{definition}

We term the constrained version of layout adjustment problem as CLADS.
We now show how to reduce a given instance of CLADS into an instance of
CUHS. We define \Seg as the set of all line segments induced by points
in $\il(S)$. 
\begin{restatable}{lemma}{forward}
   \label{cl:forward}
   If there is a set of lines \Lin containing at
   most $r$ horizontal lines and at most $c$ vertical lines that hit all
   segments in \Seg, then there is a disjoint layout \layp that has the same
   orthogonality as \il and whose height and width is bounded by $h +
   \varepsilon$ and $w + \varepsilon$, for any $\varepsilon > 0$. 
   Here $h = r + 1,\; w = c + 1$. 
\end{restatable}
To solve LADS by multiple iterations of a procedure for solving CUHS, it would
be useful to guess the width of a disjoint layout with near-optimal area. The
following observation allows us to restrict our attention to layouts with near
integral width. That makes it possible to discretize LADS, by solving a
constrained version of LADS for all values of widths in $\{1, 2, \ldots, |S|\}$. 

\begin{restatable}{lemma}{discrete}
   \label{lem:disc}
   Any disjoint layout \lay can be modified into a disjoint layout \layp having
   the same height and orthogonal ordering as \lay, such that $W(\layp) (\leq
   W(\lay))$ lies in the interval $[w, w + \varepsilon]$, where $w \in \{1, 2,
   \ldots,n\}$ and $\varepsilon > 0$ is an arbitrarily small constant.
\end{restatable}
We can similarly modify a disjoint layout \lay into an orthogonal order
preserving disjoint layout \layp which has the same width, and whose height 
lies in the interval $[h, h + \varepsilon]$ for some integer $h > 0$. Thus,
combining the two methods, we obtain the following corollary:
\begin{cor}
   \label{cor:disc}
   Any disjoint layout \lay can be modified into an orthogonal order preserving
   disjoint layout \layp, such that $W(\layp) (\leq W(\lay))$ lies in the interval
   $[w, w + \varepsilon]$ and $H(\layp) (\leq H(\lay))$ lies in the interval $[h,
   h + \varepsilon]$, where $w, h \in \{1, 2, \ldots,n\}$ and $\varepsilon > 0$
   is an arbitrarily small constant.    
\end{cor}

\vspace{-0.4cm}
\begin{restatable}{lemma}{reverse}
   \label{cl:reverse}
   For any $\varepsilon < 1/2$, if there is a disjoint layout \layp that has the
   same orthogonality as \il and whose height and width is bounded by $h +
   \varepsilon$ and $w + \varepsilon$ respectively, where $h, w$ are positive
   integers, then there is a set of lines \Lin that
   hit all segments in \Seg, containing at most $c$ vertical lines and $r$
   horizontal lines. Here $r = h - 1,\; c = w - 1$. 
\end{restatable}
All proofs of lemmas in this section are in \Cref{app:claims}.
\Cref{cl:forward,cor:disc,cl:reverse} show the close connection between CLADS
and CUHS. In subsequent sections, we exploit this connection to derive  hardness
results for CLADS.

\section{Inapproximability of Layout Adjustment Problems}
\vspace{-0.2cm}
\label{sec:APXH}
In this section, we prove APX-hardness of various layout adjustment problems. We
consider a variant of LADS where instead of minimizing the area, we would like
to minimize the perimeter of the output layout. We prove an inapproximability
result for this problem which readily follows from APX-hardness of the Uniform
Hitting Set problem. We also show that the decision problem Constrained LADS 
(CLADS) is \NP-hard. Recall that in this problem, given an initial layout of $n$
unit squares, positive integers $w, h$, and a constant $\varepsilon > 0$, the
goal is to determine if there is an orthogonal order preserving layout having
height and width at most $h + \varepsilon$ and $w + \varepsilon$ respectively.
To be precise we show a more general inapproximability result for this problem.
We prove that, given an instance of CLADS, it is \NP-hard to determine whether
there is an output layout of height and width at most $h + \varepsilon$ and $w +
\varepsilon$ respectively, or there is no output layout of respective height and
width at most $(1+\con)(h + \varepsilon)$ and $(1+\con)(w + \varepsilon)$ for
some $0 < \con < 1$. This result follows from the connection of CLADS with
Constrained Uniform Hitting Set (CUHS) described in
\Cref{sec:reduction} and APX-hardness of CUHS. The APX-hardness of CUHS follows
from the APX-hardness of UHS, to which we turn next.

\paragraph{APX-Hardness of Hitting Set Problem.}
We consider the optimization version of UHS, in which given a set of points $P$,
the goal is to find minimum number of vertical and horizontal lines that hit all
segments in $\Seg(P)$. In this section, we prove that there is no polynomial
time $(1+\con)$-factor approximation algorithm for UHS, unless \PNP, for some $0
< \con < 1$. Note that the UHS problem we consider here is a special case of the
hitting set problem where, given any set of segments $S$, the goal is to find a
hitting set for $S$. This problem is known to be \NP-hard. But, in case of UHS,
given a set of points, we need to hit all the segments induced by the points.
Thus the nontriviality in our result is to show that even this special case of
hitting set is not only \NP-hard, but also hard to approximate. To prove the
result we reduce a version of maximum satisfiability problem (5-OCC-MAX-3SAT) to
UHS. 5-OCC-MAX-3SAT is defined as follows. Given a set $X$ of $n$ boolean
variables and a conjunction $\phi$ of $m$ clauses such that each clause contains
precisely three distinct literals and each variable is contained in exactly five
clauses ($m=\frac{5n}{3}$), the goal is to find a binary assignment of the
variables in $X$ so that the maximum number of clauses of $\phi$ are satisfied.
The following theorem follows from the work of~\CCite{Feige98}.

\begin{theorem}
\label{th:max3sat}
For some $\gamma > 0$, it is \NP-hard to distinguish between an instance of
5-OCC-MAX-3SAT consisting of all satisfiable clauses, and one in which less than
$(1-\gamma)$-fraction of the clauses can be satisfied.
\end{theorem}

The crux of the hardness result is to show the existence of a reduction from 5-OCC-MAX-3SAT to UHS having
the following properties:

\begin{enumerate}
 \item Any instance of 5-OCC-MAX-3SAT in which all the clauses can be satisfied, is
     reduced to an instance of UHS in which the line
     segments in $\Seg(P)$ can be hit using at most $k$ lines, where $k$ is a function
     of $m$ and $n$. 
 \item Any instance of 5-OCC-MAX-3SAT in which less than $1-\delta$ (for 
     $0< \delta \leq 1$) fraction of the clauses can be satisfied, is reduced to an
     instance of UHS in which more than
     $(1+\apconst\delta)k$ lines are needed to hit the segments in $\Seg(P)$.
\end{enumerate}

The complete reduction appears in \Cref{App:reduc}. The next theorem follows
from the existence of such a reduction and from \Cref{th:max3sat}.

\begin{restatable}{theorem}{inapproxHS}
\label{th:inapproxHS}
There is no polynomial time $(1+\con)$-factor approximation algorithm for
UHS with $\con \leq \apconst\gamma$, unless
\PNP, $\gamma$ being the constant in \Cref{th:max3sat}.
\end{restatable}

Now we consider the variant of LADS where we would like to minimize the
perimeter $2(w+v)$ of the output layout, where $w$ and $v$ are the width and
height of the layout respectively. We refer to this problem as Layout Adjustment
for Disjoint Squares - Minimum Perimeter (LADS-MP). We note that in UHS we
minimize the sum of the number of horizontal and vertical lines ($k=r+c$). Thus
by \Cref{cl:forward} and \Cref{cl:reverse} it follows that a solution for UHS
gives a solution for LADS-MP (within an additive constant) and vice versa. Hence
the following theorem easily follows from \Cref{th:inapproxHS}.

\begin{theorem}\label{th:perihs}
No polynomial time $(1+\con')$-factor approximation algorithm exists for
LADS-MP with $\con'=\frac{\con}{4}$, unless
\PNP, $\con$ being the constant in \Cref{th:inapproxHS}.
\end{theorem}

\paragraph{Inapproximability of CUHS.}
We show that if there is a polynomial time approximate
decision algorithm for Constrained Uniform Hitting Set - Decision Problem
(CUHS), then \PNP. We use the inapproximability result of UHS for this purpose.
See \Cref{def:CUHS} for the definition of CUHS.  Now we have the following
theorem whose proof follows from \Cref{th:inapproxHS} and is given in
\Cref{App:cuhs}.

\begin{restatable}{theorem}{thcuhs}
\label{th:cuhs}
Suppose there is a polynomial time algorithm that, given $\Seg(P)$ and
non-negative integers $r,c$ as input to CUHS, 
\begin{enumerate}[(1)]
 \item outputs ``yes'', if there is a set with at most $c$ vertical and $r$
     horizontal lines that hits the segments in $\Seg(P)$; and
 \item outputs ``no'', if there is no hitting set for $\Seg(P)$ using at most
     $(1+\con)c$ vertical and $(1+\con)r$ horizontal lines, where $\con$ is the
     constant in \Cref{th:inapproxHS}. 
\end{enumerate}
Then \PNP.
\end{restatable}

\vspace{-0.4cm}
\paragraph{Inapproximability of CLADS.}
We show that the existence of a polynomial time approximate decision algorithm
for CLADS implies \PNP. See \Cref{def:CLADS} for the definition of CLADS.  Now
we have the following theorem whose proof follows from \Cref{th:cuhs} and is given
in \Cref{App:CLADS}.

\begin{restatable}{theorem}{thinapCLADS}
\label{th:inap_CLADS}
Suppose there is a polynomial time algorithm that, given $S$, \il, $w, h$, and
$\varepsilon$ as input to CLADS, 
\begin{enumerate}[(1)]
 \item outputs ``yes'', if there is an output layout \layp with $H(\layp) \leq h
     + \varepsilon$ and $W(\layp) \leq w + \varepsilon$; and
 \item outputs ``no'', if there is no output layout \layp with $H(\layp) \leq
     (1+\con')(h + \varepsilon)$ and $W(\layp) \leq (1+\con')(w + \varepsilon)$,
     where $\con'=\frac{\con}{4}$ and $\con$ is the constant in
     \Cref{th:inapproxHS}. 
\end{enumerate}
Then \PNP.
\end{restatable}

\vspace{-0.5cm}
\section{Approximation Algorithm}
\label{sec:approx}
\vspace{-0.2cm}
In this section, we describe an approximation algorithm for LADR i.e.  for a set
$R$ of axis-parallel rectangles having initial layout \il, we need to find a
disjoint layout of minimum area that preserves the orthogonal ordering of \il.
Let $W_{\max} = \max\{w_r \mid r \in R\}$ and $H_{\max} = \max\{h_r \mid r
\in R\}$ be the maximum width and maximum height, respectively, amongst all rectangles in $R$.
\Cref{lem:disc} showed that if the input consists of a set of squares $S$,
any disjoint layout of $S$ can be modified into a disjoint layout having same
orthogonality such that its width is arbitrarily close to an integer from the set
$\{1,\dots,|S|\}$. It can be seen that \Cref{lem:disc} can be extended in a
straightforward manner for a set of axis-parallel rectangles $R$ i.e.\ any
disjoint layout of $R$ can be modified into a disjoint orthogonal-order
preserving layout having a width that is arbitrarily close to an integer from the set
$\{W_{\max}, W_{\max} + 1, \dots ,W_R\}$, where $W_R = \sum\limits_{r \in R}w_r$.
We henceforth state \Cref{cor:disc} in the context of LADR as follows.

\begin{cor}
   \label{cor:gendisc}
   Let $W_R = \sum\limits_{r \in R}w_r$ and $H_R = \sum\limits_{r \in R}h_r$ be
   the sum of the widths and sum of the heights of all the rectangles in $R$,
   respectively. Then, any disjoint layout \lay of $R$ can be modified into an
   orthogonal order preserving layout \layp of $R$, such that $W(\layp) (\leq
   W(\lay))$ lies in the interval $[w, w + \varepsilon]$ and $H(\layp) (\leq
   H(\lay))$ lies in the interval $[h, h + \varepsilon]$, where $w \in
   \{W_{\max}, W_{\max} + 1, \dots,W_R\},\; h \in \{H_{\max}, H_{\max} + 1,
   \dots, H_R \}$ and $\varepsilon > 0$ is an arbitrarily small constant.
\end{cor}

Using \Cref{cor:gendisc}, we know that for any disjoint layout \lay of $R$,
there is a corresponding disjoint layout \layp having the same orthogonal order
as \layp, whose height and width are arbitrarily close to an integer from a
known set of integers. Hence, we look at all disjoint orthogonality preserving
layouts in that range, and choose the one with the minimum area as our solution.

Given positive integers $w \in \{W_{\max}, W_{\max} + 1, \dots,W_R\}$,  
$h \in \{H_{\max}, H_{\max} + 1, \dots, H_R \}$ and $0 < \varepsilon < 1$,
we formulate as a LP the problem of whether there is an orthogonal order
preserving layout \lay with $W(\lay) \leq w + \varepsilon, H(\lay) \leq h +
\varepsilon$. Recall that a layout \lay assigns a location $\lay(r) = (x_r,
y_r)$ for the center of each rectangle $r \in R$. The variables of our linear
program are $\cup_{r \in R} \{x_r, y_r \}$.  For any two rectangles $r, r' \in
R$, $\il_x(r) < \il_x(r')$ implies that $x_r < x_r'$.
We add such a constraint for each pair of rectangles in $R$, both for
$x$-coordinate and $y$-coordinate of the layout. Similarly, we add the
constraint $x_r = x_r'$ for all pair of rectangles $r,r' \in R$ for which
$\il_x(r) = \il_x(r')$. These constraints ensure orthogonality is preserved in
the output layout. 

We now look at constraints that ensure disjointness of the output layout. Let
$r$ and $r'$ be two rectangles in the initial layout \il, having dimensions
$(w_r, h_r)$ and $(w_{r'}, h_{r'})$ respectively. We define $w(r, r') =
\frac{w_r + w_{r'}}{2}$ and $h(r,r') = \frac{h_r + h_{r'}}{2}$. Let
$x_{\texttt{diff}}(r, r') = 
     \begin{cases}
          x_r - x_{r'}, & \mbox{if } \ilx(r') \leq \ilx(r) \\
          x_{r'} - x_r, & \mbox{otherwise} 
       \end{cases}$.
We define $y_{\texttt{diff}}(r, r')$ analogously. If $r, r'$ are disjoint in
some layout, then either their $x$-projections or their $y$-projections are
disjoint in that layout. Equivalently, either the difference in $x$-coordinates
of the centers of rectangles $r, r'$ is at least $w(r,r')$, or the difference in
$y$-coordinates of the centers is at least $h(r, r')$. We thus get the following
LP:

\vspace{-0.3cm}
\begin{align}
    x_r & < x_{r'}  & \forall r,r' \in R: \ilx(r) < \ilx(r')
    \label[ineq]{ineq:gxless} \\ 
    x_r & = x_{r'}  & \forall r,r' \in R: \ilx(r) = \ilx(r')
    \label[ineq]{ineq:gxeq}\\ 
    y_r & < y_{r'}  & \forall r,r' \in R: \ily(r) < \ily(r')
    \label[ineq]{ineq:gyless}\\ 
    y_r & = y_{r'}  & \forall r,r' \in R: \ily(r) = \ily(r')
    \label[ineq]{ineq:gyeq}\\ 
    \left (x_{r'} + \frac{w_{r'}}{2} \right ) - 
      \left (x_r  - \frac{w_r}{2}    \right ) & 
      \leq w + \varepsilon & \forall r,r' \in R:
      \ilx(r) < \ilx(r') \label[ineq]{ineq:gwidth}\\ 
    \left (y_{r'} + \frac{h_{r'}}{2} \right ) -
      \left (y_r  - \frac{h_r}{2} \right )  & 
      \leq h + \varepsilon & \forall r,r' \in R:
    \ily(r) < \ily(r') \label[ineq]{ineq:gheight} \\ 
    \frac{x_{\texttt{diff}}(r, r')}{w(r,r')}   + 
    \frac{y_{\texttt{diff}}(r, r')}{h(r,r')}  & \geq 1 
    & \forall r,r' \in R \label[ineq]{ineq:gdisj} 
\end{align}

\Crefrange{ineq:gxless}{ineq:gyeq} model the orthogonal ordering requirement for
a layout, while \Crefrange{ineq:gwidth}{ineq:gheight} restrict the width and
height of the layout respectively. Since any two rectangles $r, r'$ in a disjoint layout
are separated by at least half the sum of their widths in the $x$-direction
($w(r,r')$) or at least half the sum of their heights in the $y$-direction
($h(r, r')$), \Cref{ineq:gdisj} ensures that every such layout is a valid
solution for the linear program. We incorporate the linear program into
\Cref{alg:LADRSolver} for solving LADR.

\begin{algorithm}[hbt]
    \caption{$ApproxLADR(R, \il)$}
    \label{alg:LADRSolver}
    \begin{algorithmic}[1]
    \setcounter{ALC@unique}{0} 
        \REQUIRE A set of rectangles $R$, and initial layout \il.
        \ENSURE  A disjoint layout that has the same orthogonal order as \il.
        \FOR {$w = W_{\max}$ to $W_R$}
            \FOR {$h = H_{\max}$ to $H_R$}
                \IF {LP stated in \Crefrange{ineq:gxless}{ineq:gdisj} is feasible}
                   \STATE $\lay_{w,h} \leftarrow $ Layout returned by solution
                   of LP. \label{lin:feasible}
                       \label{lin:lplay}
                       \IF {$\lay^{\min}$ is undefined \OR $A(\lay_{w,h}) < A(\lay^{\min})$}
                       \STATE $\lay^{\min} \leftarrow \lay_{w,h}$
                   \ENDIF
                \ENDIF
            \ENDFOR
        \ENDFOR
        \STATE Define $\lay(R) = 2 \cdot \lay_{\min}(R)
        \text{ i.e. } \lay(r) = \left( 2*\lay^{\min}_x(r),\,
            2*\lay^{\min}_y(r) \right),\, \forall r \in R$ \label{lin:round}
        \RETURN The layout $\lay$.
    \end{algorithmic}
\end{algorithm}

\begin{lemma}
    \label{thm:LADRRatio}
    $ApproxLADR(R, \il)$ returns a $4$-approximation for LADR.
\end{lemma}
\begin{proof}
    Let $\lay_{w,h}$ be any feasible layout returned by the LP in
    \Cref{lin:feasible}, for some value of $w, h$. Let $r, r'$ be two rectangles
    in $R$, and assume that $\ilx(r) > \ilx(r')$, $\ily(r) > \ily(r')$. (The
    other cases are symmetric).  By \Cref{ineq:gdisj}, either
    $\frac{x_{\texttt{diff}}(r, r')}{w(r,r')} \geq \frac{1}{2}$ or
    $\frac{y_{\texttt{diff}}(r, r')}{h(r,r')} \geq \frac{1}{2}$. Without loss of
    generality, assume its the former. Consider the layout $\lay = 2
    \lay_{w,h}$, as in \Cref{lin:round}.
    Hence, our assumption that $\frac{x_{\texttt{diff}}(r, r')}{w(r,r')} \geq
    \frac{1}{2}$ implies that $\lay_x(r) - \lay_x(r') = 2x_r - 2x_{r'} \geq
    w(r,r')$, which satisfies the criteria for disjointness in
    \Cref{ineq:disjoint}. Since the final layout \lay returned by the algorithm
    equals $2\cdot \lay_{w',h'}$ for some feasible layout $\lay_{w',h'}$, \lay
    is a disjoint layout that also satisfies the constraints for orthogonal
    ordering in \Crefrange{ineq:gxless}{ineq:gyeq}.

    Let \lays be any disjoint layout preserving the orthogonal ordering of \il.
    We may assume, by \Cref{cor:gendisc}, that its width is at most $w' +
    \varepsilon$ and its height is at most $h' + \varepsilon$, for some integers
    $w' \in \{W_{\max}, W_{\max} + 1, \dots,W_R\},\; h' \in \{H_{\max},
    H_{\max} + 1, \dots, H_R \}$ and some $\varepsilon > 0$. Consider the
    iteration of the inner for loop in \Cref{alg:LADRSolver} with $w = w'$ and
    $h = h'$.  Since \lays is a valid solution for the LP, the
    layout $\lay_{w,h}$ computed in \Cref{lin:lplay} (and hence $\lay^{\min}$)
    has an area that is less than or equal to that of \lays.  The algorithm
    \Cref{alg:LADRSolver} returns a layout $\lay(S)$ obtained by multiplying
    each of the coordinates in $\lay^{\min}$ by a factor of 2. Hence, the layout
    $\lay(S)$ has at most twice the width and at most twice the height of \lays,
    ensuring that $A(\lay) \leq 4 * A(\lays)$.
\end{proof}

We note that since $W_R, H_R$ are not polynomial in the input size, the
resultant algorithm is a pseudo-polynomial time algorithm. But by searching
across exponentially increasing value of widths, and thereby losing a small
approximation factor, we can obtain a $4(1 + o(1))$ polynomial time
approximation for LADR. We conclude by summarizing our result as follows: 

\begin{theorem}
   \label{lem:LADRApprox}
   There is a polynomial time algorithm that returns a $4(1 +
   o(1))$-approximation for LADR i.e.\ given a set of rectangles $R$ and an
   initial layout \il, it returns an orthogonal order preserving disjoint layout
   whose area is at most $4(1 + o(1))$ times the area attainable by any such
   layout.
\end{theorem}
We note that our approach can also be used to get a $2(1 + o(1))$ approximation
for the problem of finding a layout of rectangles that minimizes the perimeter. 

\bibliographystyle{splncsnat}
\bibliography{iladr}

\begin{thebibliography}{22}
\providecommand{\natexlab}[1]{#1}
\providecommand{\url}[1]{\texttt{#1}}
\providecommand{\urlprefix}{}

\bibitem[{Adamaszek and Wiese(2015)}]{AdamaszekW15}
Adamaszek, A., Wiese, A.: A quasi-ptas for the two-dimensional geometric
  knapsack problem.
\newblock In: SODA. pp. 1491--1505 (2015)

\bibitem[{Agarwal et~al.(1998)Agarwal, van Kreveld, and Suri}]{AgarwalKS98}
Agarwal, P.K., van Kreveld, M.J., Suri, S.: Label placement by maximum
  independent set in rectangles.
\newblock Comput. Geom. 11(3-4), 209--218 (1998)

\bibitem[{Brandes and Pampel(2013)}]{BrandesP13}
Brandes, U., Pampel, B.: Orthogonal-ordering constraints are tough.
\newblock J. Graph Algorithms Appl. 17(1), 1--10 (2013)

\bibitem[{Bridgeman and Tamassia(1998)}]{Bridgeman98}
Bridgeman, S., Tamassia, R.: Difference metrics for interactive orthogonal
  graph drawing algorithms.
\newblock In: Graph Drawing, pp. 57--71 (1998)

\bibitem[{Bridgeman and Tamassia(2002)}]{Bridgeman02}
Bridgeman, S., Tamassia, R.: A user study in similarity measures for graph
  drawing.
\newblock J. Graph Algorithms Appl. 6(3), 225--254 (2002)

\bibitem[{Eades et~al.(1991)Eades, Lai, Misue, and
  Sugiyama}]{eades1991preserving}
Eades, P., Lai, W., Misue, K., Sugiyama, K.: Preserving the mental map of a
  diagram.
\newblock International Institute for Advanced Study of Social Information
  Science, Fujitsu Limited (1991)

\bibitem[{Feige(1998)}]{Feige98}
Feige, U.: A threshold of ln {\it n} for approximating set cover.
\newblock J. ACM 45(4), 634--652 (1998)

\bibitem[{Gaur et~al.(2002)Gaur, Ibaraki, and Krishnamurti}]{GaurIK02}
Gaur, D.R., Ibaraki, T., Krishnamurti, R.: Constant ratio approximation
  algorithms for the rectangle stabbing problem and the rectilinear
  partitioning problem.
\newblock J. Algorithms 43(1), 138--152 (2002)

\bibitem[{Harren et~al.(2014)Harren, Jansen, Pr{\"a}del, and van
  Stee}]{HarrenJPS14}
Harren, R., Jansen, K., Pr{\"a}del, L., van Stee, R.: A (5/3 +
  {$\varepsilon$})-approximation for strip packing.
\newblock Comput. Geom. 47(2), 248--267 (2014)

\bibitem[{Harren and van Stee(2009)}]{HarrenS09}
Harren, R., van Stee, R.: Improved absolute approximation ratios for
  two-dimensional packing problems.
\newblock In: APPROX-RANDOM. pp. 177--189 (2009)

\bibitem[{Hassin and Megiddo(1991)}]{HassinM91}
Hassin, R., Megiddo, N.: Approximation algorithms for hitting objects with
  straight lines.
\newblock Discrete Appl. Math. 30(1), 29--42 (Jan 1991)

\bibitem[{Hayashi et~al.(1998)Hayashi, Inoue, Masuzawa, and
  Fujiwara}]{hayashi98}
Hayashi, K., Inoue, M., Masuzawa, T., Fujiwara, H.: A layout adjustment problem
  for disjoint rectangles preserving orthogonal order.
\newblock In: Graph Drawing. pp. 183--197 (1998)

\bibitem[{Herman et~al.(2000)Herman, Melan{\c{c}}on, and Marshall}]{herman00}
Herman, I., Melan{\c{c}}on, G., Marshall, M.S.: Graph visualization and
  navigation in information visualization: A survey.
\newblock Visualization and Computer Graphics, IEEE Transactions on 6(1),
  24--43 (2000)

\bibitem[{Huang et~al.(2007)Huang, Lai, Sajeev, and Gao}]{Huang2007}
Huang, X., Lai, W., Sajeev, A., Gao, J.: A new algorithm for removing node
  overlapping in graph visualization.
\newblock Information Sciences 177(14), 2821--2844 (2007)

\bibitem[{van Kreveld et~al.(1999)van Kreveld, Strijk, and Wolff}]{KreveldSW99}
van Kreveld, M.J., Strijk, T., Wolff, A.: Point labeling with sliding labels.
\newblock Comput. Geom. 13(1), 21--47 (1999)

\bibitem[{Li et~al.(2005)Li, Eades, and Nikolov}]{LiEN05}
Li, W., Eades, P., Nikolov, N.S.: Using spring algorithms to remove node
  overlapping.
\newblock In: APVIS. CRPIT, vol.~45, pp. 131--140 (2005)

\bibitem[{Lodi et~al.(2002)Lodi, Martello, and Monaci}]{LodiMM02}
Lodi, A., Martello, S., Monaci, M.: Two-dimensional packing problems: A survey.
\newblock European Journal of Operational Research 141(2), 241--252 (2002)

\bibitem[{Misue et~al.(1995)Misue, Eades, Lai, and Sugiyama}]{misue95}
Misue, K., Eades, P., Lai, W., Sugiyama, K.: Layout adjustment and the mental
  map.
\newblock Journal of visual languages and computing 6(2), 183--210 (1995)

\bibitem[{Schiermeyer(1994)}]{Schiermeyer94}
Schiermeyer, I.: Reverse-fit: A 2-optimal algorithm for packing rectangles.
\newblock In: ESA. pp. 290--299 (1994)

\bibitem[{Steinberg(1997)}]{Steinberg97}
Steinberg, A.: A strip-packing algorithm with absolute performance bound 2.
\newblock SIAM J. Comput. 26(2), 401--409 (1997)

\bibitem[{Storey and M{\"u}ller(1996)}]{storey96}
Storey, M.A.D., M{\"u}ller, H.A.: Graph layout adjustment strategies.
\newblock In: Graph Drawing. pp. 487--499 (1996)

\bibitem[{Wolff(2009)}]{MapLabel}
Wolff, A.: The {Map-Labeling} bibliography.
\newblock
  \url{http://i11www.iti.uni-karlsruhe.de/~awolff/map-labeling/bibliography/maplab\_date.html}
  (2009), accessed: 2014-08-14

\end{thebibliography}
\appendix
\section{Appendix}
\subsection{Proofs for \Cref{sec:reduction}}
\label{app:claims}

\forward*
\begin{proof}
   For any $p \in P$, we define $\rho(p): P \rightarrow \Int^2$ as $\rho(p) =
   (i, j)$ where $i$ is the number of vertical lines in $L$ to the left of $p$,
   and $j$ is the number of horizontal lines in $L$ below $p$. We note that the
   function $\rho: P \rightarrow \Int^2$ is one-to-one, as otherwise, the
   segment corresponding to two points $p$ and $p'$ with $\rho(p) = \rho(p')$
   would not be hit by any line in \Lin. For convenience, if $\rho(p) = (i,j)$,
   we denote $\rho_x(p) = i$ and $\rho_y(p) = j$.

   We first consider the placement $\lay(p) = (i + \frac{1}{2}, j +
   \frac{1}{2})$, where $ i = \rho_x(p)$ and $j = \rho_y(p)$. This ensures
   disjointness, but not orthogonality - within a single column (or row) in \lay,
   there could be orthogonality violations w.r.t \il due to all endpoints having
   the same $x$ (or $y$) coordinate. This is fixed as follows.

   We define $P_x(i) = \{ p \in P \mid \rho_x(p) = i \}$, which represents the
   $i$-th column of the hitting set. We similarly define $P_y(j) = \{ p \in P
   \mid \rho_y(p) = j \}$.  Let $\alpha_i(p)$ denote the $x$-rank of $p$ within
   $P_x(i)$, and similarly let $\beta_j(p)$ denote the $y$-rank of $p$ within
   $P_y(j)$. For some $\delta > 0$, we define the layout \layp for a point $p$
   with $\rho(p) = (i, j)$ as:
   \begin{equation}
       \label{eq:layout}
       \layp(p) = \left ( i(1 + \delta) + \frac{1}{2} + \delta \cdot \frac{\alpha_i(p)}{n}, 
                  \ j(1 + \delta) + \frac{1}{2} + \delta \cdot \frac{\beta_j(p)}{n} \right )
   \end{equation}

   We observe that the maximum value of $\alpha_i(p)$ (or $\beta_j(p)$) is $n$.
   Thus, for any 2 points $p, p'$ in consecutive columns i.e.\ 
   $p \in P_x(i), p' \in P_x(i+1)$, $|\layp_x(p) - \layp_x(p')| \geq 1$,
   making the corresponding unit squares disjoint. Any 2 points in consecutive rows
   are similarly placed more than unit distance apart, which establishes the
   disjointness of the layout.

   We now show that \il and \layp have the same orthogonal ordering. Let $p, p'
   \in P$ be two unit squares and assume $ \il_x(p) \leq \il_x(p') $.
   If $ \rho_x(p) < \rho_x(p')$, it is clear that $ \layp_x(p)
   < \layp_x(p') $ as desired. If $ \rho_x(p) = \rho_x(p') = i$, then it is easy
   to see that $\il_x(p) < \il_x(p') \Rightarrow \alpha_i(p) < \alpha_i(p')
   \Rightarrow \layp_x(p) < \layp_x(p')$ and $\il_x(p) = \il_x(p') \Rightarrow
   \alpha_i(p) = \alpha_i(p') \Rightarrow \layp_x(p) = \layp_x(p')$. Similar
   reasoning applies to the $y$-coordinates.

   Since $L$ has at most $c$ vertical lines ($c \leq n - 1$), for any $p, p' \in
   P$, $|\layp_x(p) - \layp_x(p')| \leq c(1 + \delta) + \delta = c +
   (c+1)\delta$. Since this is the maximum difference of
   $x$-coordinates between any two centers of unit squares in the layout,
   $W(\layp) \leq c + 1 + (c+1)\delta = w + (c + 1)\delta$. Putting the value of
   $\delta = \frac{\varepsilon}{n}$, we have $W(\layp) \leq w + \varepsilon$.
   Similarly, we can show that $H(\layp) \leq h + \varepsilon$, thus concluding
   the proof.
\end{proof}

\discrete*
\begin{proof}
   Let $s_i$ denote the set of squares in \lay with $x$-rank $i$, and let
   $\delta > 0$ be a parameter. We create a modified layout \layp  by
   ``compressing'' the layout \lay horizontally while keeping the
   $y$-coordinates untouched. We set $\layp_x(s_1) = 1$, and assign the
   $x$-coordinates of the remaining squares in ascending order of their
   $x$-ranks in \lay. We claim that it is possible to place squares in $s_i$ 
   such that $\layp_x(s_i) \in [b, b + (i - 1) \cdot \delta]$, where $0 \leq  b
   \leq i - 1 \in \Intp$.

   For the base case, we place the squares in $s_2$. If the projection of $s_2$
   on the $y$-axis is not disjoint from that of $s_1$, then set $\layp_x(s_2)
   = \layp_x(s_1) + 1$. If not, set $\layp_x(s_2) = \layp_x(s_1) +
   \delta$. In both cases, $s_1$ and $s_2$ do not satisfy \Cref{ineq:disjoint}
   in layout \layp, and hence are disjoint in \layp.

   Assume inductively that for $j \leq i - 1$, we have defined $\layp_x(s_j)$
   such that $\layp_x(s_j) \in [b, b + (j - 1) \cdot \delta]$ for some $0 \leq b
   \leq j-1$. We now place $s_i$. Set $\layp_x(s_i) = \infty$ initially.
   Consider the $y$-projections of $s_i$ and $s_{i-1}$ in \lay. If they are not
   disjoint, then set $\layp_x(s_i) = \layp_x(s_{i-1}) + 1$.  If not, $s_i$ can
   be translated to the left till either $\layp_x(s_i) = \layp_x(s_{i-1}) +
   \delta$ or it touches some square in $s_{i'}, i' < i$, whichever happens
   first. In the former case, $\layp_x(s_i) \in [b, b + (i-1) \cdot \delta]$ and
   $b \leq i-1$. In the latter case, $s_{i'}$ satisfies the induction hypothesis
   and hence $\layp_x(s_i) ( = \layp_x(s_{i'}) + 1)$ is in the interval $[b' +
   1, b' + 1 + (i'-1) \cdot \delta]$, where $b' \leq i' - 1 \leq i - 1$. Hence,
   in all cases, $\layp_x(s_i)$ satisfies the induction claim, and hence we
   claim it is true for all $i \leq n$. 

   It remains to bound the total width of the layout \layp. Since $\layp_x(s_1)
   = 1$ and $\layp_x(s_i) \leq b + i \cdot \delta$ for some integer $b \leq i -
   1$, it follows that the width of the layout is in $[b, b + n \cdot \delta]$
   for some integer $b$. Setting $\delta = \frac{\varepsilon}{n}$ concludes the
   proof. 
\end{proof}

\reverse*
\begin{proof}
   Assume without loss of generality that \layp is contained in a rectangle $[0,
   w + \varepsilon] \times [0, h + \varepsilon]$ in the plane. We consider the
   following set of lines:
   \[ x = i + \frac{1}{2} + \delta, \; \text{for } i = 1,2,\dots, w-1 \qquad
      \text{and} \qquad 
    y = j + \frac{1}{2} + \delta, \; \text{for } j = 1,2,\dots, h-1. \]  
   We pick a $\delta < \frac{1}{2}$ so that none of these lines contain a rectangle
   center. We now show that the resultant set of lines $L$ constitute a hitting set
   for \Seg. 

   Let $\seg \in \Seg$ be an arbitrary line segment, whose endpoints $p, p'$ are
   the centers of unit squares $s, s'$ in the disjoint layout \layp. Since $s,
   s'$ do not intersect, then either $|\layp_x(s) - \layp_x(s')| \geq 1$ or
   $|\layp_y(s) - \layp_y(s')| \geq 1$. We assume without loss of generality its
   the former. We know that $0 \leq \layp_x(s), \layp_x(s') \leq w + \varepsilon$ and
   none of the vertical lines in $L$ pass through $p, p'$. Since successive
   vertical lines are unit distance apart, there must be at
   least one vertical line $ x = \frac{i'}{2} + \delta$ that lies between $p,
   p'$. We can argue similarly if $|\layp_y(s) - \layp_y(s')| \geq 1$ using the
   set of horizontal lines in $L$. Thus, \Seg is hit by at least one line in
   $L$, making the latter a hitting set consisting of at most $c = w - 1$
   vertical lines and $r = h - 1$ horizontal lines.  
\end{proof}

\subsection{Proof of \Cref{th:inapproxHS}}
\label{App:uhs}

\inapproxHS*
\begin{proof}
Suppose there is a polynomial time $(1+\con)$-factor approximation algorithm for
UHS for some $\con \leq \apconst\gamma$. Now consider the reduction $\pi$. For
the instance of UHS obtained from an instance of 5-OCC-MAX-3SAT in which all the
clauses are satisfied, there is a hitting set of size at most $k$; so the
approximation algorithm finds a hitting set of size at most $(1+\con)k$. For any
instance obtained from an instance of 5-OCC-MAX-3SAT in which less than
$1-\gamma$ fraction of the clauses are satisfied, needs more than
$(1+\apconst\gamma)k \geq (1+\con)k$ lines. Thus using this algorithm we can
distinguish between an instance of 5-OCC-MAX-3SAT consisting of the clauses all
of which can be satisfied, and one in which less than $1-\gamma$ fraction of the
clauses can be satisfied. Hence from \Cref{th:max3sat} it follows that \PNP,
which completes the proof of this theorem.
\end{proof}

\subsection{Reduction of 5-OCC-MAX-3SAT into Hitting Set Problem}\label{App:reduc}
A set \Lin of horizontal and vertical lines is said to \textit{separate} a set
$P$ of points, if for each point $p\in P$, there is a 2-dimensional cell
(possibly unbounded) in the arrangement of the lines in \Lin that contains $p$
and no other points of $P$.

Recall that for any set of points $P$, $\Seg(P)$ is the set of segments induced
by $P$. Suppose a set \Lin of horizontal and vertical lines hits all the
segments of $\Seg(P)$. By definition a line hits a segment, if it passes through
the interior of the segment, but does not intersect either endpoint of the
segment. Thus if we perturb each line in \Lin that passes through a segment
endpoint, then no line in \Lin intersect any segment endpoint and \Lin still
remains a hitting set for $\Seg(P)$. Henceforth, by a hitting set of lines \Lin
for $\Seg(P)$ we mean the lines in \Lin hit the segments in $\Seg(P)$, but no
line in \Lin intersects any point of $P$. Now our claim is that if \Lin is a
hitting set for $\Seg(P)$, \Lin separates $P$. If not, then there exists a cell
containing two points, and the corresponding segment is not being hit by any
line, which cannot be true. Conversely, if \Lin separates the points of $P$,
then all the segments of $\Seg(P)$ are being hit. The notion of separation will
help us to simplify our arguments for bounding the minimum number of lines
required to hit all the segments induced by a point set. Thus from now onwards,
we use the notion of hitting and separation interchangeably.

For any set of segments $S$ and a hitting set of lines $L$ for $S' \supseteq S$,
consider the subset of $L$ consisting of every line that hits at least one
segment of $L$; denote its cardinality by $N_L(S)$. For any set of segments $S$,
denote the minimum number of lines needed to hit all segments of $S$ by $N(S)$.
Given a set of lines \Lin, a point set $P_1$ is said to be separated from
another point set $P_2$, if each segment $(p,p')$ is hit by some line in
\Lin, where $p \in P_1, p' \in P_2$. For any two point sets $P_1$ and $P_2$,
denote the minimum number of lines needed to separate $P_1$ from $P_2$ by
$N(P_1,P_2)$.

For any set $P$ of points, define $x$-span of $P$ to be the interval
$[x_{\min},x_{\max}]$ on real line, where $x_{\max}$ and $x_{\min}$ are the
maximum and minimum among the $x$-coordinates of the points of $P$. Similarly,
define $y$-span of $P$ corresponding to the $y$-coordinates of its points. Now
we have the following lemma. 

\begin{lemma}\label{lem:span}
Suppose $P_1,\ldots,P_l$ are point sets with pairwise disjoint $x$-spans and
pairwise disjoint $y$-spans. Then for any hitting set of lines $L$ for a set of
segments $S\supseteq \bigcup_{i=1}^l \Seg(P_i)$, $$N_L(\bigcup_{i=1}^l
\Seg(P_i))=\sum_{i=1}^l N_L(\Seg(P_i))$$
\end{lemma}

\begin{proof}
Let $L_i$ be the subset of $L$ consisting of every line that hits at least one
segment in $\Seg(P_i)$ for $i=1,\ldots,l$. Note that it is sufficient to show
that for any $i,j \in \{1,\ldots,l\}$ such that $i\neq j$, $L_i\cap L_j=\phi$.
Consider any horizontal line $y=a$ of $L_i$. $a$ must lie in $y$-span of $P_1$.
As $y$-spans of $P_i$ and $P_j$ are disjoint, this line cannot hit any segment
of $\Seg(P_j)$, and thus cannot belong to $L_j$. Similarly, any horizontal line
of $L_j$ cannot hit any segment of $\Seg(P_i)$, and thus cannot belong to $L_i$.
Now consider any vertical line $x=b$ of $L_i$. $b$ must lie in $x$-span of
$P_i$. As $x$-spans of $P_i$ and $P_j$ are disjoint this line cannot hit any
segment of $\Seg(P_j)$, and thus cannot belong to $L_j$. In a similar way, any
vertical line of $L_j$ cannot hit any segment of $\Seg(P_i)$, and thus cannot
belong to $L_i$. Hence $L_i\cap L_j=\phi$.
\end{proof}

\begin{figure}
    \centering
    \hspace{-40mm}
    \begin{subfigure}[b]{0.3\textwidth}
        \centering
        \includegraphics[height=40mm]{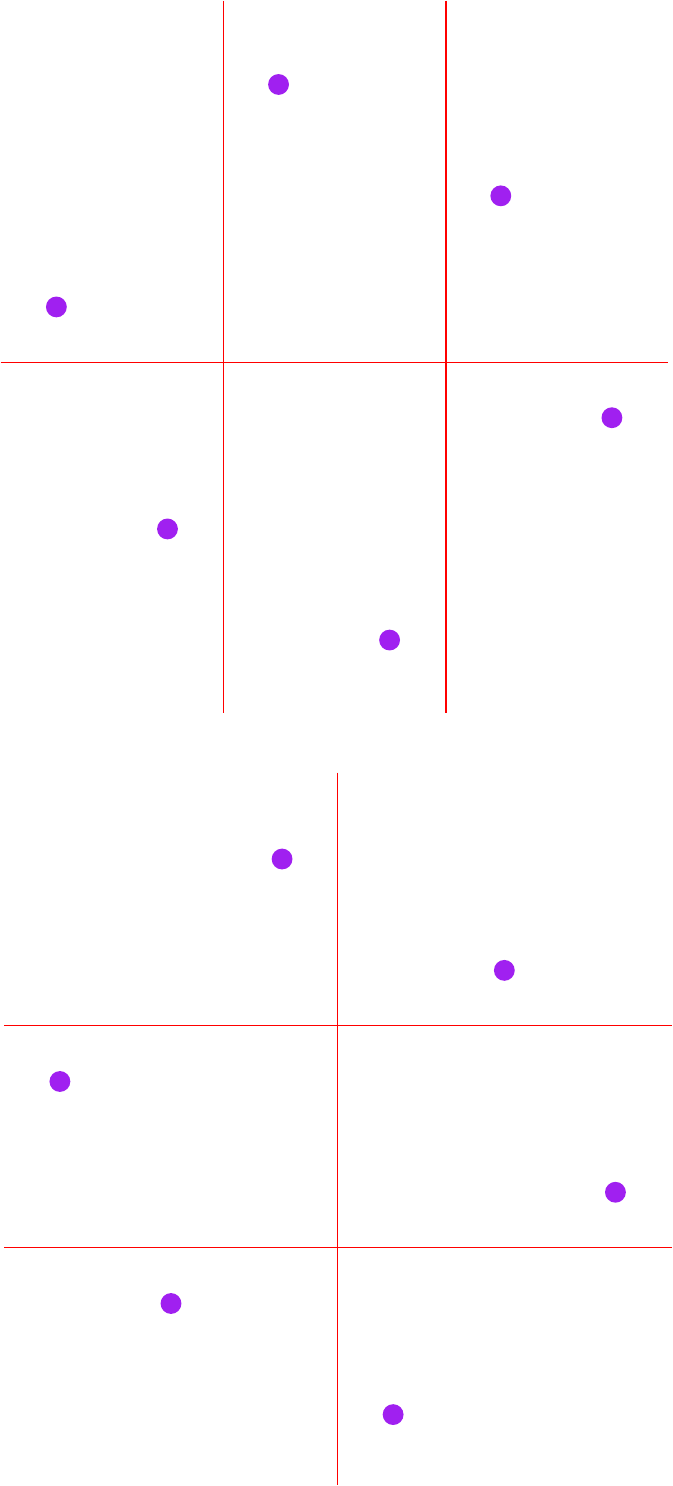}
        \caption{$ $}
        \label{fig:fig1}
    \end{subfigure}
    \hspace{1mm}
    \begin{subfigure}[b]{0.3\textwidth}
        \centering
        \includegraphics[height=40mm]{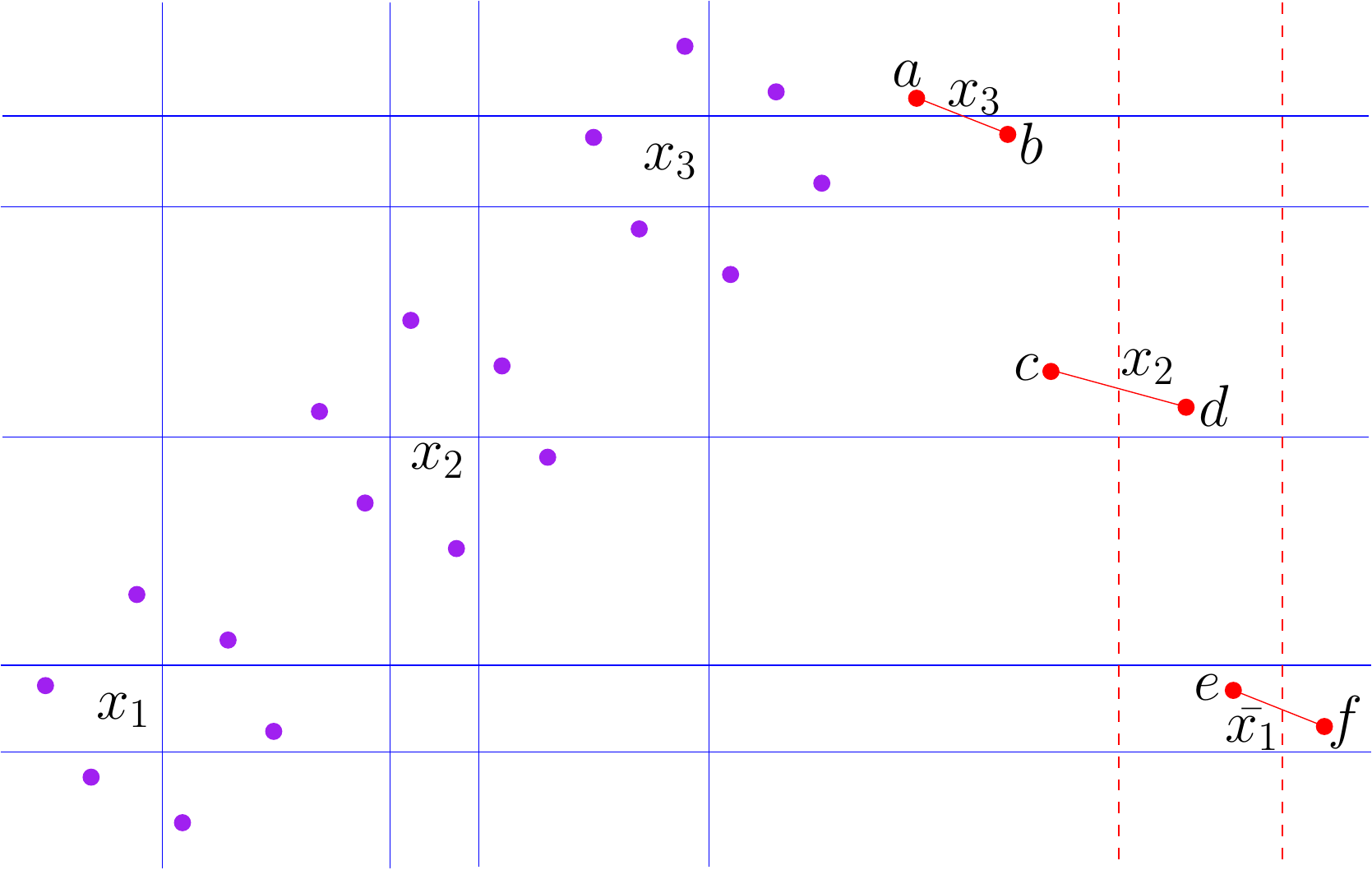}
        \caption{$ $}
        \label{fig:clause}
    \end{subfigure}
    \caption{\textit{(a) Optimal configurations of hitting set for variable
        ($x_i$) points. The top one corresponds to $x_i=0$ and the bottom one
        corresponds to $x_i=1$. (b) The points corresponding to the clause $x_3
        \lor x_2 \lor \bar{x_1}$. The hitting set for variable gadgets
    corresponds to the assignment $x_1=1, x_2=0, x_3=1$}}
    \label{fig:short}
\end{figure}

Now we discuss the construction of an instance of UHS
from a given instance of 5-OCC-MAX-3SAT. We construct a gadget that consists of
three classes of points - variable points, clause points and auxilliary points.
\paragraph*{Variable Points} For each variable we take a translate of the six
points in $I=\{(1,4),(2,2)$, $(3,6), (4,1),(5,5),(6,3)\}$. We note that three
lines are necessary and sufficient to separate these six points. To be precise
there are two optimal choices one involving a horizontal line $y=a$, where $4 <
a < 5$ and the other involving a horizontal line $y=b$, where $3 < b < 4$. See
\Cref{fig:fig1} for an illustration. A set of points $U'$ is called
$(p,q)$-translate of a point set $U$, if $U'$ is obtained by translating the $x$
and $y$-coordinates of the points in $U$ by $p$ and $q$ respectively. For each
variable $x_i$ we take a set $V_i$ which is $((6(i-1),6(i-1))$-translate of $I$,
where $1\leq i\leq n$. Let $V=\cup_{i=1}^n V_i$. 

We note that the $x$-spans (resp. $y$-spans) of the sets $V_i$ for $1\leq i\leq
n$ are pairwise disjoint. With respect to $V_i$ for $1\leq i\leq n$, let us call
a horizontal line a Type 0 (resp. Type 1) line, if it has the form $y=b$ for
$6(i-1)+3 < b < 6(i-1)+4$ (resp. $y=a$ for $6(i-1)+4 < a < 6(i-1)+5$). Now we
have the following observation.

\begin{obs}\label{obs:var}
Three lines are necessary and sufficient to hit all the segments induced by
$V_i$, i.e $N(\Seg(V_i))=3$, where $1\leq i\leq n$. In addition, there are
exactly two types of optimal hitting sets for $\Seg(V_i)$ - one that uses a Type
1 line, but no Type 0 line, and the other that uses a Type 0 line, but no Type 1
line.
\end{obs}

We will associate the first choice of hitting set for $\Seg(V_i)$ with the
assignment $x_i=1$ and the second choice with $x_i=0$. Now we have the following
lemma.

\begin{lemma}\label{lem:var}
$4n-1$ lines are sufficient to hit all segments of $\Seg(V)$, where the segments
of $\Seg(V_i)$ are being hit using 3 lines for $1\leq i\leq n$.
\end{lemma}

\begin{proof}
Note that separating $V$ involves separating $V_i$ from $V_{i+1}$ for each
$1\leq i\leq n-1$, and separating points of $V_i$ for each $1\leq i\leq n$. By
\Cref{obs:var} three lines are sufficient to separate the points of $V_i$. Also
$N(V_i,V_{i+1})=1$, as any vertical or horizontal line between $V_i$ and
$V_{i+1}$ separates one from the other. Thus in total $3n+n-1=4n-1$ lines are
sufficient.
\end{proof}

Given any hitting set for $\Seg(V)$ with at most $4n-1$ lines we want to
construct an assignment of the boolean formula $\phi$. To be precise we want to
ensure that always 3 lines are used to separate $V_i$, if at most $4n-1$ lines
are used to separate $V$. This ensures that one of the two configurations of 3
lines are used to separate $V_i$. Then we can assign binary values to the
variables depending on the configuration used. \Cref{fig:cntrexmpl} shows an
example where $4n-1$ lines are used to separate $V$, but $4$ lines are used to
separate one of the $V_i$.
\begin{figure*}[h]
 \centering
    \begin{minipage}[c]{0.4\textwidth}
    \centering
  \includegraphics[width=60mm]
    {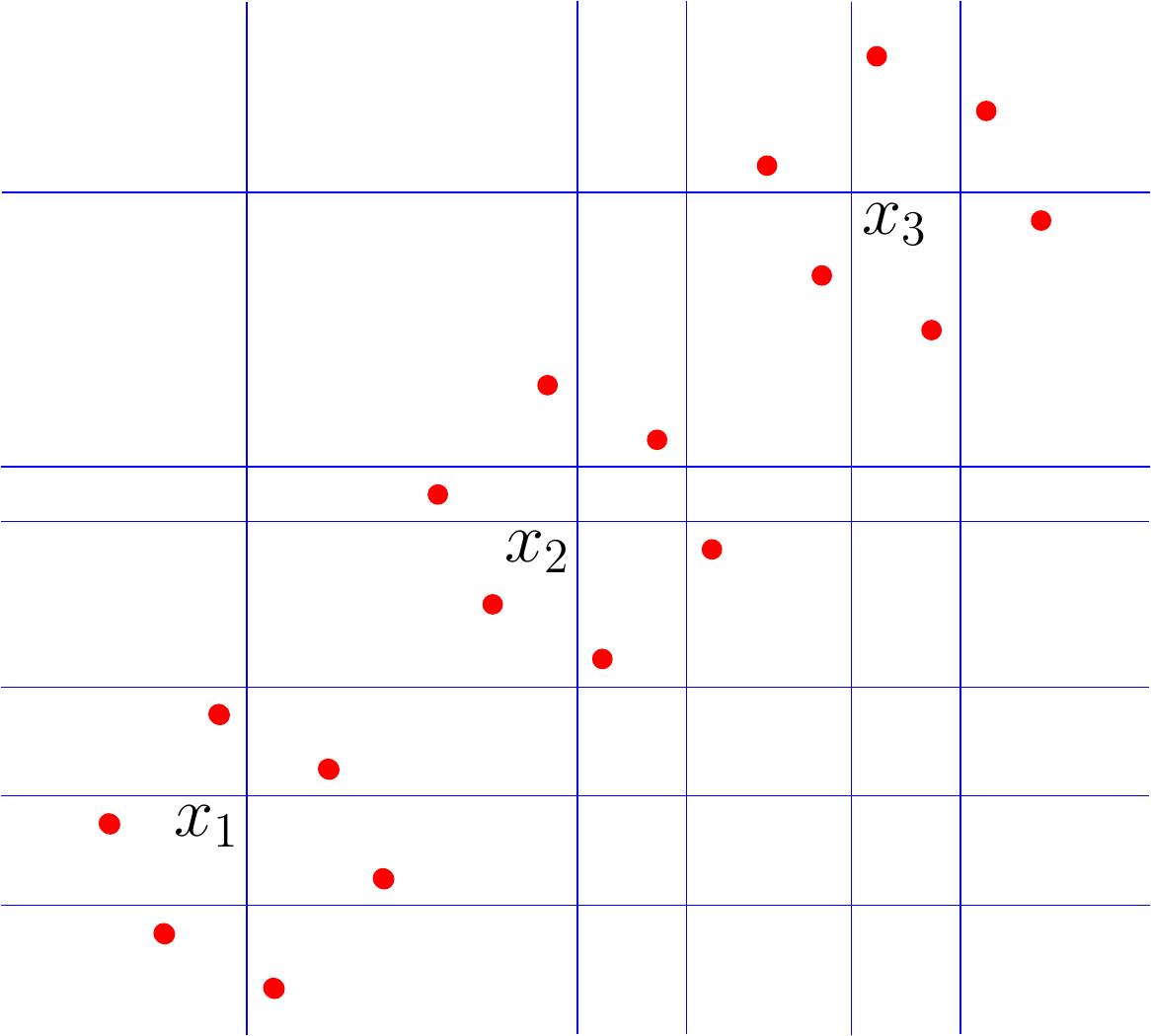}\\     
  \end{minipage}    
  \caption{\textit{11 lines are used to separate $V$, but 4 lines are used to separate $V_2$}}
  \label{fig:cntrexmpl}  
\end{figure*}
To ensure exactly 3 lines are used to separate each $V_i$, we can make sure that
$n-1$ lines between $V_i$ and $V_{i+1}$ are always picked up, for $1\leq i\leq
n-1$. Then as 3 lines are needed to separate each $V_i$, exactly 3 lines are
used for each $V_i$, if $4n-1$ lines are used in total. Later we show that a
subset of the auxilliary points will enforce this constraint.

\paragraph*{Clause Points} For each clause $C_j$ we take a set $T_j$ of ten
points that is a
union of two sets $T_j^1$ and $T_j^2$. $T_j^1$ consists of a translate of the
four points in $J=\{(2,-1), (5,-2), (6,-3), (9,-4)\}$. To be precise $T_j^1$ is
$((6n+1+10(j-1),4(j-m))$-translate of $J$. In \Cref{fig:fig2} $p,q,r$ and $s$
are the points of $T_j^1$ corresponding to the clause $x_3 \lor x_2 \lor
\bar{x_1}$. The points in $T_j^2$ depend on the literals of $C_j$. Let $u_1,
u_2$ and $u_3$ are the literals of $C_j$. For each literal we take two points.
The $x$-coordinates of the points corresponding to $u_1, u_2$ and $u_3$ are
$\{6n+1+10(j-1)+1, 6n+1+10(j-1)+3\}, \{6n+1+10(j-1)+4, 6n+1+10(j-1)+7\}$ and
$\{6n+1+10(j-1)+8, 6n+1+10(j-1)+10\}$ respectively. The $y$-coordinates of these
points depend on the form of the literals. Let $u_t$ is corresponding to the
variable $x_i$ for $t \in \{1,2,3\}$. If $u_t=x_i$, the $y$-coordinates of the
two points corresponding to $u_t$ are $6(i-1)+5-\epsilon j$ and
$6(i-1)+4+\epsilon j$ respectively, where $\epsilon$ is a very small positive
number. Otherwise, $u_t=\bar{x_i}$ and the $y$-coordinates are
$6(i-1)+4-\epsilon j$ and $6(i-1)+3+\epsilon j$ respectively. Let
$T^1=\cup_{j=1}^m T_j^1$, $T^2=\cup_{j=1}^m T_j^2$ and $T=T^1\cup T^2$.

\begin{figure}[ht]
\centering
\includegraphics[height=70mm]{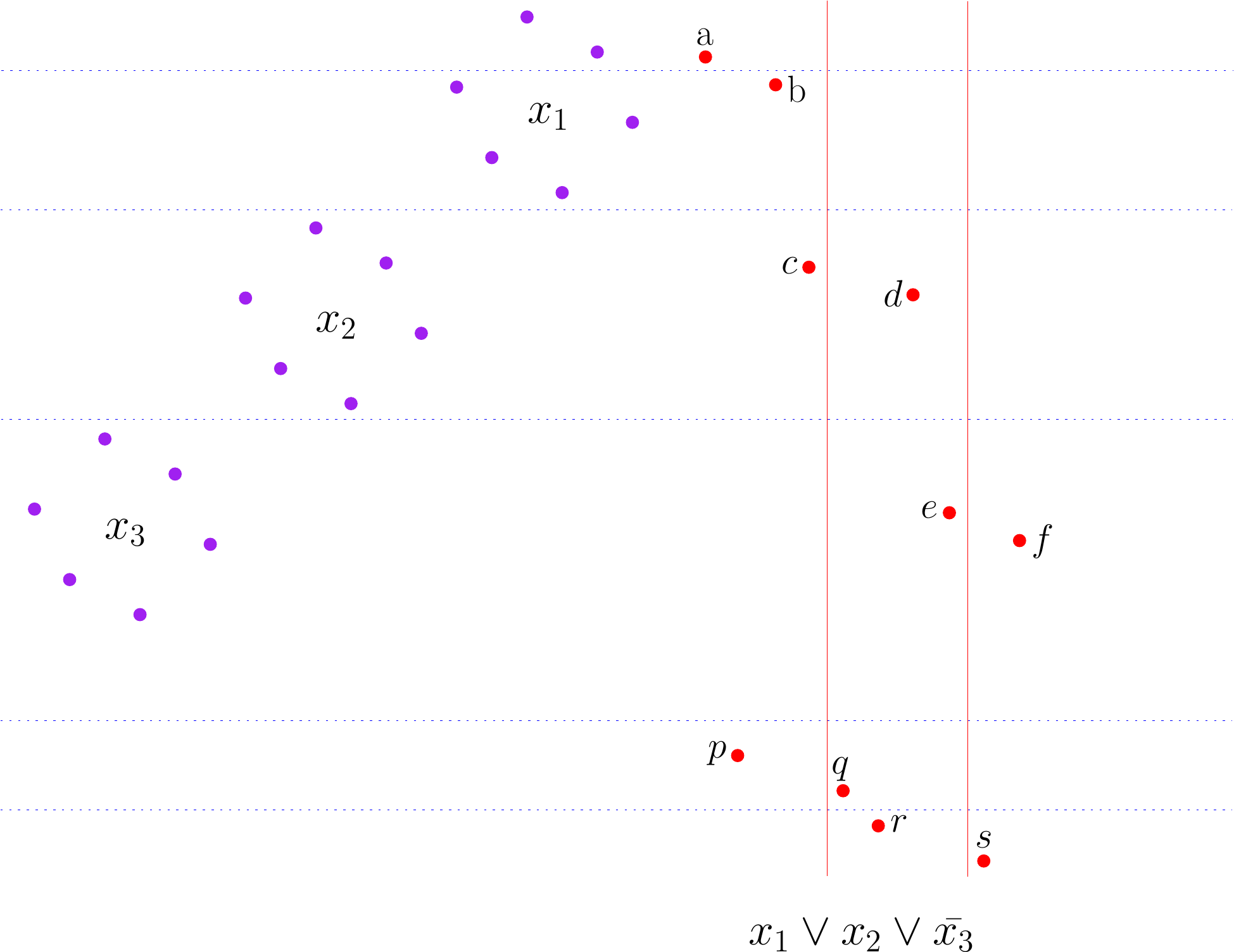}
\caption{\textit{The points corresponding to the clause $x_3 \lor x_2 \lor
\bar{x_1}$ and a scheme to separate them}} 
\label{fig:fig2}
\end{figure}

In \Cref{fig:fig2} the points corresponding to the literals $x_3$, $x_2$ and
$\bar{x_1}$ are $\{a,b\}$, $\{c,d\}$ and $\{e,f\}$ respectively. Consider the
points of $T_j$ in decreasing order of $y$-coordinates. We divide the points in
this order into groups of two.  Thus the first three groups are corresponding to
$T_j^2$ and the remaining two groups are corresponding to $T_j^1$. We denote the
first three groups by $T_{j,s}^2$ for $1\leq s\leq 3$, and the remaining two
groups by $T_{j,t}^1$ for $1\leq t\leq 2$. We note that the $y$-spans of all
these five groups are pairwise disjoint. Also $x$-spans of the three groups of
$T_j^2$ are pairwise disjoint. Similarly, $x$-spans of the two groups of $T_j^1$
are pairwise disjoint. But the $x$-spans of $T_{j,1}^1$ and $T_{j,1}^2$, and
$T_{j,1}^1$ and $T_{j,2}^2$ have nonempty intersection. Thus it is possible to
separate the points in $T_{j,1}^1$ as well as $T_{j,1}^2$ (resp. $T_{j,1}^1$ as
well as $T_{j,2}^2$) using a single vertical line. Similarly, it is possible to
separate the points in $T_{j,2}^1$ as well as $T_{j,2}^2$ (resp. $T_{j,2}^1$ as
well as $T_{j,3}^2$) using a single vertical line. See \Cref{fig:fig2} for an
illustration. We note that $x$-spans (resp. $y$-spans) of the sets $T_{j,t}^1$
for $1\leq t\leq 2$ and $1\leq j\leq m$ are pairwise disjoint. 

While constructing the points for $C_1$ we make sure that 
a literal corresponding to $V_1$ is never chosen as the first literal of $C_1$.
Note that this does not violate any generality, as each clause contains three
distinct literals. Later this will help us to argue about the number of lines
necessary to separate all the points. 

Now we discuss a scheme to separate the points of $T_j$ using horizontal and
vertical lines. The scheme considers the points in groups of two as defined
before. At first we will see how to separate the groups from each other and then
we will separate the points in the individual groups. Note that four horizontal
lines are sufficient to separate these five groups of points from each other
(see \Cref{fig:fig2}). Also suppose the points in at least one group
corresponding to the literals are separated by one additional horizontal line.
Then two lines are necessary to separate the remaining points of $T_j$ and two
vertical lines are sufficient for that purpose. Hence we have the following
observation.

\begin{obs}\label{obs:clause}
Suppose the five groups in $T_j$ as defined above are separated from each other
by four horizontal lines and points in at least one of the groups of $T_j^2$ are
separated by a horizontal line, then two more lines are necessary to separate
all the points of $T_j$ and two vertical lines are sufficient for that purpose.
\end{obs}

Suppose that as a precondition we ensure that the five groups of $T_j$ are
separated from each other for each $1\leq j\leq m$. Furthermore, suppose that we
have fixed an optimal hitting set type for each $V_i$ for $1\leq i\leq n$. If
the corresponding assignment satisfies the clause $C_j$, we can argue that the
hitting sets for the $\Seg(V_i)$ have a horizontal line that separates the
points in at least one of the three groups of $T_j^2$. Then using
\Cref{obs:clause} two additional vertical lines suffice to separate the points
in $T_j$. If the corresponding assignment does not satisfy $C_j$, then the
points in none of the three groups of $T_j^2$ are separated by the optimal
hitting sets for the $\Seg(V_i)$. Then at least three more lines ae needed to
separate the points in $T_j$. 

Now we discuss how the preconditions could be satisfied. The three groups of
$T_j^2$ can be separated from each other by the horizontal lines which are used
to separate $V_i$ and $V_{i+1}$ for $1\leq i\leq n$. An additional horizontal
line would be sufficient to separate $T_j^1$ from $T_j^2$ for all $1\leq j\leq
m$. Now if a vertical line separates the two groups of $T_j^1$ from each other, then it
also separates the points in the middle group of $T_j^2$. Thus we cannot argue
that two more lines are needed to separate the points of $T_j$. Hence we have to
ensure that the two groups of $T_j^1$ are separated using a horizontal line. 
We will enforce this constraint by a subset of auxilliary points.

\paragraph*{Auxilliary points} The set of auxilliary points $A$ consists of five
point sets $\{A^1,\ldots,A^5\}$. $A^1$ and $A^2$ ensure that $T_{j,1}^1$ and
$T_{j,2}^1$ are separated by a horizontal line for $1\leq j\leq m$. $A^3$
ensures that the points in each $V_i$ are separated using exactly 3 lines, when
$4n-1$ lines are used to separate the points in $V$. $A^4$ ensures that $T_j$ is
separated from $T_{j+1}$ by a vertical line for each $1\leq j\leq m$. $A^5$
consists of four points $(6n+.5,.5),(6n+.5,-.5),(6n+1.5,.5)$ and $(6n+1.5,-.5)$.
Note that 2 lines are necessary and sufficient to separate these four points.
Moreover, a horizontal and a vertical line that separate these points also
separate $T^1$ from $T^2$, and $V$ from $T$ respectively (see \Cref{fig:auxi}). 

$A^1$ is composed of $m$ point sets each of which consists of a translate of the
four points in
$F=\{(0,-\epsilon),(-1,-1+\epsilon),(-2,-2\epsilon),(-3,-1+2\epsilon)\}$, where
$\epsilon>0$ is a very small number. $A^1=\cup_{j=1}^{m} A_j^1$, where each
$A_j^1$ is $(-4m+3+4(j-1),-4m+2+4(j-1))$-translate of $F$ (see \Cref{fig:auxi}). 
Note that two lines are necessary, and one horizontal and one vertical line are
sufficient to separate the points in each $A_j^1$. By construction, the
horizontal line corresponding to $A_j^1$ also separates the two groups of
$T_j^1$ from each other (see \Cref{fig:auxi}). 

$A^2$ is composed of $m-1$ point sets each of which consists of a translate of
the four points in
$H=\{(-\epsilon,0),(-1+\epsilon,-1),(-2\epsilon,-2),(-1+2\epsilon,-3)\}$.
$A^2=\cup_{l=1}^{m-1} A_l^2$, where each $A_l^2$ is
$(-4m+4l,-8m+7+4(l-1))$-translate of $H$ (see \Cref{fig:auxi}).  Note that one
horizontal and one vertical line are sufficient to separate the points of each
$A_l^2$. Moreover, the vertical line corresponding to $A_l^2$ also separates the
two groups $A_{l}^1$ and $A_{l+1}^1$ for $1\leq l\leq m-1$. And, any vertical
line that hits a segment in $\Seg(A_l^1)$ separates $A_{l-1}^2$ from $A_{l}^2$
for $2\leq l\leq m-1$. 

$A^3$ consists of $n-1$ point sets each of which is a translate of the points in
$F$. $A^3=\cup_{l=1}^{n-1} A_l^3$, where $A_l^3$ is
$(-4m-1-4(l-1),7+6(l-1))$-translate of $F$. Thus the points in $A^3$ lie towards
north-west of the points in $A^1$ (see \Cref{fig:auxi}). Note that one
horizontal and one vertical line are sufficient to separate the points in each
$A_l^3$. Also the $y$-span of $A_l^3$ is between the $y$-spans of $V_l$ and
$V_{l+1}$ for each $1\leq l\leq n-1$. Thus the horizontal line corresponding to
$A_l^3$ also separates $V_l$ from $V_{l+1}$ for $1\leq l\leq n-1$.

$A^4$ consists of $m-1$ point sets each of which is a translate of the points in
$H$. $A^4=\cup_{j=1}^{m-1} A_j^4$, where $A_j^4$ is
$(-8m+3-4(j-1),6n+2+10j)$-translate of $H$. Note that two lines are necessary,
and one horizontal and one vertical line are sufficient to separate the points
in each $A_j^4$. In addition, the $x$-span of $A_j^4$ is between the $x$-spans
of $T_j$ and $T_{j+1}$ for each $1\leq j\leq m-1$. Thus the vertical line
corresponding to $A_j^4$ also separates $T_j$ from $T_{j+1}$ for $1\leq j\leq
m-1$ (see \Cref{fig:auxi}). 

Now consider the collection of sets $A_l^1$ for $1\leq l\leq m$, $A_s^2$ for
$1\leq s\leq m-1$, $A_t^3$ for $1\leq t\leq n-1$, $A_j^4$ for $1\leq j\leq m-1$,
$T_{j,1}^1$ for $1\leq j\leq m$ and $T_{j,2}^1$ for $1\leq j\leq m$. By
construction, $x$-spans (resp. $y$-spans) of any two sets in this collection are
disjoint. Also $y$-spans (resp. $x$-spans) of the sets $V_i$ for $1\leq i\leq n$
and $A_t^3$ for $1\leq t\leq n-1$ are pairwise disjoint. Similarly, $y$-spans
(resp. $x$-spans) of the sets $T_j^1$ for $1\leq j\leq m$ and $A_j^4$ for $1\leq
j\leq m-1$ are pairwise disjoint. As $x$-spans (resp. $y$-spans) of $V$ and
$A^1\cup A^2$, $V$ and $A^4\cup A^5$, and $V$ and $\cup_{j=1}^m T_j^1$ are
disjoint, we have the following observation.

\begin{obs}\label{obs:disjoint}
Consider the collection of sets $A_t^1$ for $1\leq t\leq m$, $A_l^2$ for $1\leq
l\leq m-1$, $A_s^3$ for $1\leq s\leq n-1$, $A_j^4$ for $1\leq j\leq m-1$, $V_i$
for $1\leq i\leq n-1$, $T_{j,1}^1$ for $1\leq j\leq m$, $T_{j,2}^1$ for $1\leq
j\leq m$ and $A^5$. $x$-spans (resp. $y$-spans) of any two sets in this
collection are disjoint.
\end{obs}

Now we have the following lemma regarding the properties of the auxilliary
points.

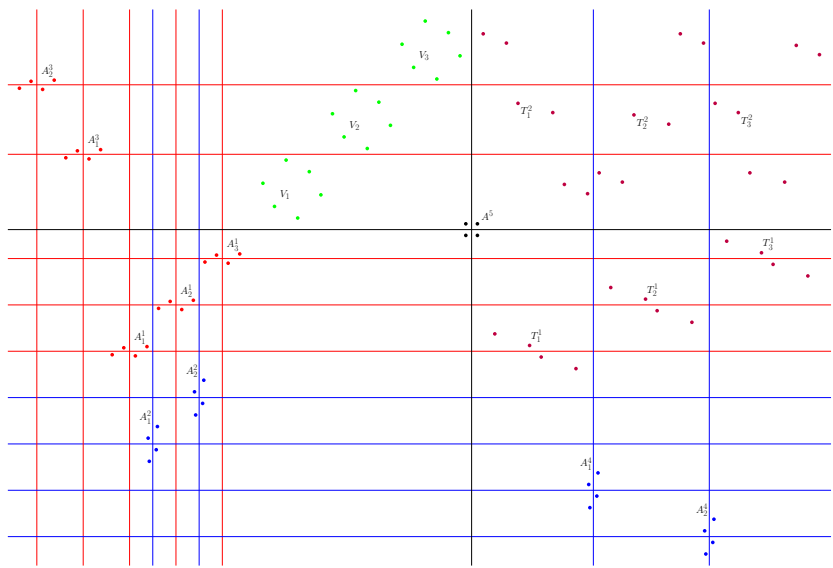
\begin{figure}[hbt]
\tikzstyle{node} = [circle, fill=black, minimum size=9pt, inner sep=0pt]
\centering
\resizebox{.9\linewidth}{!} {
\begin{tikzpicture}
\scriptsize

\pgfmathtruncatemacro{\maxy}{19};
\pgfmathtruncatemacro{\maxx}{50};
\pgfmathtruncatemacro{\miny}{-29};
\pgfmathtruncatemacro{\minx}{-21};

\foreach \i in {1,...,3}
{
        \pgfmathtruncatemacro{\p}{6*(\i-1)};
        \pgfmathtruncatemacro{\q}{6*(\i-1)};
        \pgfmathtruncatemacro{\j}{4-\i};
        \node[node,fill=green, label = {[label distance=1.3cm] 340:\Huge{$V_{\i}$}} ] at (\p+1,\q+4) {$ $};
        \node[node,fill=green, label =$ $] at (\p+2,\q+2) {$ $};
        \node[node,fill=green, label =$ $] at (\p+3,\q+6) {$ $};
        \node[node,fill=green, label =$ $] at (\p+4,\q+1) {$ $};
        \node[node,fill=green, label =$ $] at (\p+5,\q+5) {$ $};
        \node[node,fill=green, label =$ $] at (\p+6,\q+3) {$ $};        
}

\foreach \j in {1,...,3}
{
   \pgfmathtruncatemacro{\p}{ 6*3+1+10*(\j-1)};
   \pgfmathtruncatemacro{\q}{ 4*(\j-3)};
   \node[node,fill=purple, label =$ $] at (\p+2,\q-1) {$ $};
        \node[node,fill=purple, label =80:\Huge{$T_{\j}^1 $}] at (\p+5,\q-2) {$ $};
        \node[node,fill=purple, label =$ $] at (\p+6,\q-3) {$ $};
        \node[node,fill=purple, label =$ $] at (\p+9,\q-4) {$ $};
}
\pgfmathtruncatemacro{\a}{ 6*3+1+10*(1-1)};
\node[node,fill=purple, label =$ $] at (\a+1,17-1*.1) {$ $};
\node[node,fill=purple, label =$ $] at (\a+3,16+1*.1) {$ $};
\node[node,fill=purple, label =340:\Huge{$T_{1}^2 $}] at (\a+4,11-1*.1) {$ $};
\node[node,fill=purple, label =$ $] at (\a+7,10+1*.1) {$ $};
\node[node,fill=purple, label =$ $] at (\a+8,4-1*.1) {$ $};
\node[node,fill=purple, label =$ $] at (\a+10,3+1*.1) {$ $};

\pgfmathtruncatemacro{\b}{ 6*3+1+10*(2-1)};
\node[node,fill=purple, label =$ $] at (\b+1,5-1*.1) {$ $};
\node[node,fill=purple, label =$ $] at (\b+3,4+1*.1) {$ $};
\node[node,fill=purple, label =340:\Huge{$T_{2}^2 $}] at (\b+4,10-1*.1) {$ $};
\node[node,fill=purple, label =$ $] at (\b+7,9+1*.1) {$ $};
\node[node,fill=purple, label =$ $] at (\b+8,17-1*.1) {$ $};
\node[node,fill=purple, label =$ $] at (\b+10,16+1*.1) {$ $};

\pgfmathtruncatemacro{\b}{ 6*3+1+10*(3-1)};
\node[node,fill=purple, label =$ $] at (\b+1,11-1*.1) {$ $};
\node[node,fill=purple, label =340:\Huge{$T_{3}^2 $}] at (\b+3,10+1*.1) {$ $};
\node[node,fill=purple, label =$ $] at (\b+4,5-1*.1) {$ $};
\node[node,fill=purple, label =$ $] at (\b+7,4+1*.1) {$ $};
\node[node,fill=purple, label =$ $] at (\b+8,16-1*.1) {$ $};
\node[node,fill=purple, label =$ $] at (\b+10,15+1*.1) {$ $};

\foreach \j in {1,...,3}
{
        \pgfmathtruncatemacro{\p}{-4*3+3+4*(\j-1)};
        \pgfmathtruncatemacro{\q}{-4*3+2+4*(\j-1)};
        \node[node,fill=red, label =93:\Huge{$A_{\j}^1 $}] at (\p,\q-.1) {$ $};
        \node[node,fill=red, label =$ $] at (\p-1,\q-1+.1) {$ $};
        \node[node,fill=red, label =$ $] at (\p-2,\q-2*.1) {$ $};
        \node[node,fill=red, label =$ $] at (\p-3,\q-1+2*.1) {$ $};
        \draw [red] (\minx,\q-.5) -- (\maxx,\q-.5);
        \draw [red] (\p-1.5,\maxy) -- (\p-1.5,\miny);
}
\foreach \j in {1,...,2}
{
        \pgfmathtruncatemacro{\p}{-4*3)+4*(\j)};
        \pgfmathtruncatemacro{\q}{-8*3+7+4*(\j-1)};
        \node[node,fill=blue, label ={[label distance=.3cm] 160:\Huge{$A_{\j}^2 $}} ] at (\p-.1,\q) {$ $};
        \node[node,fill=blue, label =$ $] at (\p-1+.1,\q-1) {$ $};
        \node[node,fill=blue, label =$ $] at (\p-2*.1,\q-2) {$ $};
        \node[node,fill=blue, label =$ $] at (\p-1+2*.1,\q-3) {$ $};
        \draw [blue] (\minx,\q-1.5) -- (\maxx,\q-1.5);
        \draw [blue](\p-.5,\miny) -- (\p-.5,\maxy);
}
\foreach \j in {1,...,2}
{
        \pgfmathtruncatemacro{\p}{-4*3-1-4*(\j-1)};
        \pgfmathtruncatemacro{\q}{7+6*(\j-1)};
        \pgfmathtruncatemacro{\i}{3-\j};
        \node[node,fill=red, label =93:\Huge{$A_{\j}^3 $}] at (\p,\q-.1) {$ $};
        \node[node,fill=red, label =$ $] at (\p-1,\q-1+.1) {$ $};
        \node[node,fill=red, label =$ $] at (\p-2,\q-2*.1) {$ $};
        \node[node,fill=red, label =$ $] at (\p-3,\q-1+2*.1) {$ $};
        \draw [red] (\minx,\q-.5) -- (\maxx,\q-.5);
        \draw [red] (\p-1.5,\maxy) -- (\p-1.5,\miny);
}
\foreach \j in {1,...,2}
{
        \pgfmathtruncatemacro{\p}{6*3+2+10*(\j)};
        \pgfmathtruncatemacro{\q}{-8*3+3-4*(\j-1)};
        \node[node,fill=blue, label ={[label distance=.3cm] 160:\Huge{$A_{\j}^4 $}} ] at (\p-.1,\q) {$ $};
        \node[node,fill=blue, label =$ $] at (\p-1+.1,\q-1) {$ $};
        \node[node,fill=blue, label =$ $] at (\p-2*.1,\q-2) {$ $};
        \node[node,fill=blue, label =$ $] at (\p-1+2*.1,\q-3) {$ $};
        \draw [blue] (\minx,\q-1.5) -- (\maxx,\q-1.5);
        \draw [blue](\p-.5,\miny) -- (\p-.5,\maxy);
}
\node[node,fill=black, label =$ $] at (6*3+.5,.5) {$ $};
\node[node,fill=black, label =$ $] at (6*3+.5,-.5) {$ $};
\node[node,fill=black, label ={[label distance=.2cm] 45:\Huge{$A^5 $}}] at (6*3+1.5,.5) {$ $};
\node[node,fill=black, label =$ $] at (6*3+1.5,-.5) {$ $};
\draw [black] (\minx,0) -- (\maxx,0);
\draw [black] (6*3+1,\miny) -- (6*3+1,\maxy);

\end{tikzpicture}
}
\caption{\textit{Construction of points corresponding to the clauses $\{x_3\lor
   x_2\lor \bar{x_1}, x_1\lor \bar{x_2}\lor x_3, x_2\lor x_1\lor \bar{x_3}\}$
($n=3, m=3$).}}
\label{fig:auxi}
\end{figure}

\begin{lemma}\label{lem:auxi}
The points of $A$ have the following properties:
\begin{enumerate}[i)]
 \item $4m-2$ lines are necessary to hit the segments in $\Seg(A^1\cup A^2)$.
    Moreover, there is a hitting set for $\Seg(A^1\cup A^2)$ consisting of
    $4m-2$ lines that also separates the two groups of $T_j^1$ from each other
    for $1\leq j\leq m$. 
 \item $2(n-1)$ lines are necessary and sufficient to hit the segments in
    $\bigcup_{t=1}^{n-1} \Seg(A_t^3))$.
 \item $2(m-1)$ lines are necessary and sufficient to hit the segments in
    $\bigcup_{j=1}^{m-1} \Seg(A_j^4))$.
\end{enumerate}
\end{lemma}

\begin{proof}
\textit{i)} Consider any hitting set of lines $L$ for the segments in
$\Seg(A^1\cup A^2)$. At first we show that $N_L(\Seg(A^1\cup A^2))\geq 4m-2$.
Note that $N(\Seg(A_j^1))=N(\Seg(A_l^2))=2$ for $1\leq j\leq m, 1\leq l\leq
m-1$. Also $x$-spans and $y$-spans of the sets $A_j^1$ for $1\leq j\leq m$ and
$A_l^2$ for $1\leq l\leq m-1$ are pairwise disjoint. Hence by \Cref{lem:span},
$N_L(\Seg(A^1\cup A^2))\geq N_L(\bigcup_{j=1}^{m} \Seg(A_j^1) \cup
\bigcup_{l=1}^{m-1} \Seg(A_l^2))=\sum_{j=1}^{m} N_L(\Seg(A_j^1)) +
\sum_{l=1}^{m-1} N_L(\Seg(A_l^2))\geq \sum_{j=1}^{m} N(\Seg(A_j^1)) +
\sum_{l=1}^{m-1} N(\Seg(A_l^2))=2(m+m-1)=4m-2.$

Now we construct a hitting set of lines for $\Seg(A^1\cup A^2)$ having at most
$4m-2$ lines. 
We take one horizontal and one vertical line to separate the points of each of
the sets $A_j^1$ and $A_l^2$ for $1\leq j\leq m, 1\leq l\leq m-1$. We need
$4m-2$ lines for this. Now the vertical line corresponding to $A_l^2$ also
separates $A_{l}^1$ from $A_{l+1}^1$ for $1\leq l\leq m-1$. And, the vertical
line corresponding to $A_j^1$ also separates the groups $A_{j-1}^2$ and
$A_{j}^2$ for $2\leq j\leq m-1$. Thus all the points in $A^1 \cup A^2$ are
separated. By construction, the horizontal line corresponding to $A_j^1$ also
separates the two groups of $T_j^1$ from each other, and the result follows.\\\\
\textit{ii)} First we prove that $2(n-1)$ lines are necessary. Consider any
hitting set of lines $L$ for the segments in $\bigcup_{t=1}^{n-1} \Seg(A_t^3))$.
We note that $N(\Seg(A_t^3))=2$ for $1\leq t\leq n-1$. Also $x$-spans and
$y$-spans of the sets $A_t^3$ for $1\leq t\leq n-1$ are pairwise disjoint. Thus
using \Cref{lem:span}, $$N_L(\bigcup_{t=1}^{n-1} \Seg(A_t^3))=\sum_{t=1}^{n-1}
N_L(\Seg(A_t^3))\geq \sum_{t=1}^{n-1} N(\Seg(A_t^3))=2(n-1).$$

To prove the sufficient condition we construct a hitting set of lines for the
segments in $\bigcup_{t=1}^{n-1} \Seg(A_t^3)$. For each $1\leq t\leq n-1$ we
take one horizontal and one vertical line that separates the points in $A_t^3$.
Hence the $2(n-1)$ lines in total hit all the segments in $\bigcup_{t=1}^{n-1}
\Seg(A_t^3)$.\\\\ \textit{iii)} Consider any hitting set of lines $L$ for the
segments in $\bigcup_{j=1}^{m-1} \Seg(A_j^4))$. We note that $N(\Seg(A_j^4))=2$
for $1\leq j\leq m-1$. Also $x$-spans and $y$-spans of the sets $A_j^4$ for
$1\leq j\leq m-1$ are pairwise disjoint. Thus using \Cref{lem:span},
$$N_L(\bigcup_{j=1}^{m-1} \Seg(A_j^4))=\sum_{j=1}^{m-1} N_L(\Seg(A_j^4))\geq
\sum_{j=1}^{m-1} N(\Seg(A_j^4))=2(m-1).$$

Now to prove the sufficient condition we construct a hitting set for the
segments in $\bigcup_{t=1}^{n-1} \Seg(A_t^3)$ consisting of $2m-2$ lines. For
each $1\leq j\leq m-1$ we take one horizontal and one vertical line that
separates the points in $A_j^4$. Hence the $2(m-1)$ lines in total hit all the
segments in $\bigcup_{j=1}^{m-1} \Seg(A_j^4)$.
\end{proof}

Now we formally discuss the reduction from 5-OCC-MAX-3SAT to \textit{Hitting Set
Problem}. Given a 5-OCC-MAX-3SAT formula $\phi$ we construct the set of points
$V$ and $T$ corresponding to variables and clauses respectively in the way
mentioned before. We also construct $A$, the set of auxilliary points, in the
same way mentioned before. 
Let $P=V\cup T\cup A$. 
Let $k=5n+8m-4$. Now we show that the reduction has the two properties as
claimed before. The following lemma ensures the first property.

\begin{lemma}\label{lem:sattohit}
Suppose $\phi$ is satisfiable, then $N(\Seg(P))\leq k$. 
\end{lemma}

\begin{proof}
Note that it is sufficient to show the existence of a set $L$ of at most $k$
horizontal and vertical lines which separate the points of $P$. The set $L$ is
constructed in the following way. $L$ is composed of three disjoint sets of
lines $L^v, L^c$ and $L^a$. $L^v, L^c$ and $L^a$ are the sets of lines which
separate the points of $V$, $T$ and $A$ respectively. These lines also separate
the groups $V,T$ and $A$ from each other. Now we will consider the construction
of these sets.

We add the horizontal line $y=0$ and the vertical line $x=6n+1$ to $L^a$. These
two lines divide the plane into four cells. Each cell contains one point of
$A^5$ and some additional points. Thus the points in $A^5$ are separated by
these two lines. In addition to the points of $A^5$ the four cells contain the
points of four disjoint sets $A^1\cup A^2$, $V\cup A^3$, $T^1\cup A^4$, and
$T^2$ respectively. Hence the four point sets in the four cells are already
separated from each other. We show how to separate the points in each of these
four sets.

By \Cref{lem:auxi}(i) it follows that there is a hitting set for $\Seg(A^1\cup
A^2)$ consisting of $4m-2$ lines that also separates the two groups of $T_j^1$
from each other for $1\leq j\leq m$. We add such $4m-2$ lines to $L^a$. Now
$|L^a|=4m$.

We note that the vertical line correponding to $A_1^1$ separates $A^3$ from $V$.
By \Cref{lem:auxi}(ii) $2(n-1)$ lines are sufficient to hit the segments in
$\Seg(A^3)$. We add such a set of $2n-2$ lines to $L^a$. Now $|L^a|=4m+2n-2$.
Note that the horizontal line corresponding to $A_i^3$ separates $V_i$ from
$V_{i+1}$ for $1\leq i\leq n-1$. By \Cref{obs:var} three lines are sufficient to
separate the points of $V_i$ and only two configurations of lines are possible
for that purpose. For each variable $x_i$, if $x_i$ is 1, we select the
configuration with the Type 1 horizontal line $y=6(i-1)+4.5$ for $V_i$. If $x_i$
is 0, we select the configuration with the Type 0 horizontal line $y=6(i-1)+3.5$
for $V_i$. We add the 3 lines corresponding to the chosen configuration to $L^v$
for each $V_i$. Thus $|L^v|=3n$. Hence the points in $V$ are now separated. Also
at least one horizontal line has been chosen per $V_i$ and thus $A_t^3$ is
separated from $A_{t+1}^3$ for each $1\leq t\leq n-2$. Thus all the points in
$V\cup A^3$ are now being separated. Note that a vertical line corresponding to
$V_n$ and a horizontal line corresponding to $V_1$ separate $V\cup A^3$ from
$A^4$. Also the same vertical line separates $A^1\cup A^2$ from $A^4$. 

By \Cref{lem:auxi}(iii) $2(m-1)$ lines are sufficient to hit the segments in
$\Seg(A^4)$. We add such a set of $2m-2$ lines to $L^a$. Now
$|L^a|=4m+2n-2+2m-2=6m+2n-4$. Note that the vertical line corresponding to
$A_j^4$ also separates $T_j$ from $T_{j+1}$ for all $1\leq j\leq m-1$. Also the
horizontal line corresponding to $A_1^2$ separates $A^4$ from $T_1$. 

Now consider the points in $T$. The horizontal lines that separate $V_i$ from
$V_{i+1}$ for all $1 \leq i \leq n-1$, also separate the three groups
$T_{j,1}^2, T_{j,2}^2$ and $T_{j,1}^2$ from each other, for $1 \leq j \leq m$.
Also the hoizontal lines in $L^a$ corresponding to $A^1$ separate $T_{j,1}^1$
from $T_{j,2}^1$ for $1\leq j\leq m$. Now we claim that for each $1\leq j\leq m$
the points in at least one of the groups of $T_j^2$ are being separated by a
horizontal line in $L^v$. Consider any clause $C_j=u_1\lor u_2\lor u_3$. Now as
$C_j$ is satisfied there must be a satisfied literal of $C_j$. Without loss of
generality let $u_1$ be a satisfied literal of $C_j$ and it is corresponding to
the variable $x_i$. If $u_1$ is $x_i$, the points in $T_j^2$ corresponding to
$u_1$ are $(6n+1+10(j-1)+1,6(i-1)+5-\epsilon j)$ and
$(6n+1+10(j-1)+3,6(i-1)+4+\epsilon j)$. As $x_i$ is 1, previously we chose the
configuration with the Type 1 line for $V_i$. Thus these two points are already
separated. Otherwise, if $u_1$ is $\bar{x_i}$, the points in $T_j^2$
corresponding to $u_1$ are $(6n+1+10(j-1)+1,6(i-1)+4-\epsilon j)$ and
$(6n+1+10(j-1)+3,6(i-1)+3+\epsilon j)$. Now as $x_i$ is 0, previously we chose
the configuration with the Type 0 line for $V_i$. Thus in this case also the two
points are already separated. 

Consider the set of points $T_j$ for any $1\leq j\leq m$. All the five groups of
$T_j$ are separated from each other, and the points in at least one group of
$T_j^2$ are separated. Thus by \Cref{obs:clause} only two vertical lines are
sufficient to separate the points in $T_j$. We construct such a set $L^c$
containing $2m$ vertical lines corresponding to $T_j$ for all $1 \leq j \leq m$.
As the points in $T_j$ are separated for all $1\leq j\leq m$, and $T_j$ is
separated from $T_{j+1}$ for all $1\leq j\leq m$, all the points in $T$ are
separated.

Also the vertical line corresponding to $T_{1,1}^1$ and the horizontal line
corresponding to $A_m^1$ separate $A^4$ from $T^1$. Now the first literal of
$C_1$ cannot be corresponding to $V_1$. Thus the vertical line corresponding to
$T_{1,1}^1$ and a horizontal line corresponding to $V_1$ separate $A^4$ from
$T^2$. Lastly, the vertical line corresponding to $T_{j,1}^1$ separate
$A_{j-1}^1$ from $A_j^1$ for $2\leq j\leq m-1$.

We set $L=L^a\cup L^v\cup L^c$ and thus the lines of $L$ separate all the points
of $P$ by our construction. Now $|L|=(6m+2n-4)+3n+2m=5n+8m-4=k$ which completes
the proof of the lemma.
\end{proof}

Now we show that the reduction ensures the second property as well, i.e if less
than $1-\delta$ ($0< \delta \leq 1$) fraction of the clauses of $\phi$ are
satisfied, then more than $(1+\apconst\delta)k$ lines are needed to hit the
segments in $\Seg(P)$. The idea is to show that any hitting set of lines for
$\Seg(P)$ that uses exactly 3 lines to seprate the points in $V_i$, uses at
least $k$ lines. Thus if at most $k+p$ lines are used to hit all segments in
$\Seg(P)$, then more than 3 lines can be used for at most $p$ $V_i$'s. Now if 3
lines are used to separate the points in $V_i$, then we can find a binary
assignment of the variable.  Using the exact three lines configurations of at
least $n-p$ remaining $V_i$'s we can create an assignment (the values of the
other variables are chosen arbitrarily) which makes a good fraction of the
clauses to be satisfied. Before moving on now we have some definitions.

Consider any hitting set of lines $L$ for $\Seg(P)$. A variable $x_i$ is said to
be \textit{good}, if the number of lines of $L$ that hit at least one segment of
$\Seg(V_i)$ is equal to 3. A variable is said to be \textit{bad}, if it is not
good. A clause $C_j$ is said to be \textit{bad}, if it has at least one of the
following two properties, (i) any line of $L$ that hits the segment of
$\Seg(T_{j,t}^2)$, is vertical, for each $1\leq t\leq 3$, and (ii) it contains a
literal corresponding to a bad variable. The bad clauses which have the first
property are called bad clauses of first type, and the remaining bad clauses are
called bad clauses of second type.  A clause is said to be \textit{good}, if it
is not bad. Now we have the following lemma.

\begin{lemma}\label{lem:badclause}
Consider any hitting set $L$ for $\Seg(P)$ consisting of at most $k+p$ lines.
The number of bad clauses corresponding to $L$ is at most $5p$.
\end{lemma}

\begin{proof}
Let $a$ be the number of good variables. Then 

\begin{equation}\label{eqn:var}
   \sum_{i=1}^{n} N_L(\Seg(V_i)) \geq 3a+4(n-a)=4n-a
\end{equation}

As each variable appears in exactly 5 clauses the bad variables can appear in at
most $5(n-a)$ clauses. Let $b$ be the number of bad clauses of first type. Then
the total number of bad clauses is at most $5(n-a)+b$. Now consider a bad clause
$C_l$ of first type. A vertical line that separate the points in $T_{l,t}^2$
might not separate the points of any of $T_{l,1}^1$ and $T_{l,2}^1$. Let $C'$ be
the subset of the bad clauses which has one such corresponding vertical line,
and denote the cardinality of $C'$ by $b_1$. Then each of the three vertical
lines, corresponding to any of the remaining $b-b_1$ bad clauses $C_s$ separate
the points in either $T_{s,1}^1$ or $T_{s,2}^1$. Now for any clause $C_j$
$x$-spans (resp. $y$-spans) of the sets $T_{j,1}^1$ and $T_{j,2}^1$ are
disjoint. Thus at least two lines are needed to separate the points in
$T_{j,1}^1$ and $T_{j,2}^1$. This implies 

\begin{equation}\label{eqn:clause}
   \sum_{j=1}^{m} N_L(\Seg(T_{j,1}^1))+ \sum_{j=1}^{m} N_L(\Seg(T_{j,2}^1))\geq 2m+(b-b_1)
\end{equation}

By \Cref{obs:disjoint} the sets $A_t^1$ for $1\leq t\leq m$, $A_l^2$ for $1\leq
l\leq m-1$, $A_s^3$ for $1\leq s\leq n-1$, $A_j^4$ for $1\leq j\leq m-1$, $V_i$
for $1\leq i\leq n-1$, $T_{j,1}^1$ for $1\leq j\leq m$, $T_{j,2}^1$ for $1\leq
j\leq m$ and $A^5$ have pairwise disjoint $x$ spans and pairwise disjoint
$y$-spans. Applying \Cref{lem:span}, \Cref{eqn:var} and \Cref{eqn:clause} we
count a lower bound of $N_L(\Seg(P))$ which is as follows.

\begin{align*} 
    &N_L(\Seg(P)) \\ 
    &\geq N_L \Big (\bigcup_{t=1}^{m} \Seg(A_t^1)\cup
    \bigcup_{l=1}^{m-1} \Seg(A_l^2)\cup \bigcup_{s=1}^{n-1} \Seg(A_s^3)\cup
    \bigcup_{j=1}^{m-1} \Seg(A_j^4)\cup \bigcup_{i=1}^{n} \Seg(V_i)\cup \\
     &\qquad \qquad \bigcup_{j=1}^{m} \Seg(T_{j,1}^1)\cup \bigcup_{j=1}^{m} \Seg(T_{j,2}^1)\cup
    \Seg(A^5) \Big ) \\ 
    &= \sum_{t=1}^{m} N_L \Big (\Seg(A_t^1) \Big)+ \sum_{l=1}^{m-1}
    N_L \Big (\Seg(A_l^2) \Big)+ \sum_{s=1}^{n-1} N_L \Big (\Seg(A_s^3) \Big)+
    \sum_{j=1}^{m-1} N_L \Big (\Seg(A_j^4) \Big) + \\ 
    &\qquad \qquad \sum_{i=1}^{n} N_L\Big(\Seg(V_i) \Big)+ \sum_{j=1}^{m}
    N_L \Big (\Seg(T_{j,1}^1) \Big )+ \sum_{j=1}^{m} N_L \Big (\Seg(T_{j,2}^1)
    \Big)+ N_L\Big (\Seg(A^5) \Big) \\
    &\geq \sum_{t=1}^{m} N(\Seg(A_t^1))+ \sum_{l=1}^{m-1} N(\Seg(A_l^2))+
    \sum_{s=1}^{n-1} N(\Seg(A_s^3))+ \sum_{j=1}^{m-1} N(\Seg(A_j^4))+ \\
    &\qquad \qquad (4n-a)  + (2m+b-b_1)+ N(\Seg(A^5)) \\ 
    &= \sum_{t=1}^{m} 2+ \sum_{l=1}^{m-1} 2+
    \sum_{s=1}^{n-1} 2+ \sum_{j=1}^{m-1} 2+(4n-a)+ (2m+b-b_1)+2 \\ 
    &=2m+2m-2+2n-2+2m-2+4n-a+2m+b-b_1+2\\ 
    &= (8m+5n-4)+(n+b-a-b_1)= k+(n+b-a-b_1).
\end{align*}

Now consider any clause $C_j$ in $C'$. There is at least one vertical line $l_j$
in $L$ that separate the points in one of the groups of $T_j^2$, but do not
separate the points in any of $T_{j,1}^1$ and $T_{j,2}^1$. Let $l_j$ separates
the points in $T_{j,t}^2$ for some $1\leq t\leq 3$. The $x$-coordinates of the
points in $T_{j,t}^2$ are $6n+1+10(j-1)+3(t-1)+1$ and $6n+1+10(j-1)+3(t-1)+3$.
Then the equation of $l_j$ is $x=c$ for some $6n+1+10(j-1)+3(t-1)+1 < c <
6n+1+10(j-1)+3(t-1)+3$. We claim that the vertical line $l_j$ is not considered
in the previous counting of lower bound of $N_L(\Seg(P))$. If not, then one of
the groups considered in the previous counting must contain two points
$(p_x^1,p_y^1)$ and $(p_x^2,p_y^2)$ such that $p_x^1 < c$ and $p_x^2 > c$. The
maximum among the $x$-coordinates of the points in $A\setminus A_4$ is $6n+1.5 <
c$, which is actually a point of $A^5$. Thus none of the subsets of $A\setminus
A_4$ contain such points. Also the $x$-spans of $A_t^4$ and $T_j$ are disjoint
for any $1\leq t\leq m-1$. Thus no such $A_t^4$ contains those points.
Similarly, the maximum among the $x$-coordinates of the points in $V$ is $6n <
c$, and thus no $V_i$ contains such points. Now consider any $T_l^1$ for $l <
j$. Then the maximum among the $x$-coordinates of the points in $T_l^1$ is
$6n+1+10(l-1)+9 = 6n+1+10l-1 \leq 6n+1+10(j-1)-1\leq 6n+1+10(j-1)+3(t-1)+1 < c$.
Thus none of $T_{l,1}^1$ and $T_{l,2}^1$ contain such points. Now consider any
$T_l^1$ for $l > j$. Then the minimum among the $x$-coordinates of the points in
$T_l^1$ is $6n+1+10(l-1)+2 \geq 6n+1+10(j+1-1)+2 \geq 6n+1+10(j-1)+10+2 \geq
6n+1+10(j-1)+3(3-1)+6 > 6n+1+10(j-1)+3(t-1)+2 > c$. Thus in this case also none
of $T_{l,1}^1$ and $T_{l,2}^1$ contain the desired points. But this leads to a
contradiction and hence $l_j$ is not considered in the counting before.

Now $|C'|=b_1$ and thus in total there are $b_1$ additional vertical lines in
$L$. So $N_L(\Seg(P)) \geq k+(n+b-a-b_1)+b_1=k+n+b-a$. As the total number of
lines in $L$ is at most $k+p$, 

\begin{flalign*}
&k+n+b-a \leq k+p  \\\text{or, }&n-a+b \leq p
\end{flalign*}

Thus the total number of bad clauses is $5(n-a)+b \leq 5(n-a+b)\leq 5p$.
\end{proof}

\begin{lemma}\label{lem:frac_sat}
Suppose at most $k+\frac{\delta}{5} m$ ($0< \delta \leq 1$) lines are sufficient
to hit the segments in $\Seg(P)$, then there is an assignment which satisfy at
least $1-\delta$ fraction of the clauses.
\end{lemma}

\begin{proof}
Consider any hitting set of lines $L$ for $\Seg(P)$ which uses at most
$k+\frac{\delta}{5} m$ lines. We construct an assignment of the clauses having
the desired property. Let $V^1$ be the set of good variables corresponding to
$L$. Then by \Cref{obs:var} for each $x_i \in V^1$ the optimal configuration
with either a Type 1 line or a Type 0 line has been chosen. For each $x_i \in
V^1$ if the configuration with Type 1 line is chosen, set $x_i$ to 1. Otherwise,
set $x_i$ to 0. For each $x_i \in V\setminus V^1$ assign an arbitrary binary
value to $x_i$. We show that this assignment makes at least $(1-\delta)m$
clauses to be satisfied.

Consider any good clause $C_j$. All of its literals are corresponding to good
variables and the points of at most two groups of $T_j^2$ have been separated
using only vertical lines. Let $T_{j,t}^2$ be the group whose points are
separated by a horizontal line. Let $u$ and $x_i$ be the literal and variable
corresponding to this group. If $u$ is $x_i$, the $y$-coordinates of the points
in $T_{j,t}^2$ are $6(i-1)+5-\epsilon j$ and $6(i-1)+4+\epsilon j$. Thus the
horizontal line that separates these two points is of Type 1, which implies 1 is
assigned to $x_i$ and $C_j$ is true. If $u$ is $\bar{x_i}$, the $y$-coordinates
of the points in $T_{j,t}^2$ are $6(i-1)+4-\epsilon j$ and $6(i-1)+3+\epsilon
j$. Thus the horizontal line that separates the points in $T_{j,t}^2$ is of Type
0. This implies 0 is assigned to $x_i$ and in this case also $C_j$ is true. Thus
any good clause is satisfied by the assignment. From \Cref{lem:badclause} it
implies that the number of good clauses is at least $m-5\frac{\delta}{5}
m=(1-\delta)m$, which completes the proof of this lemma.
\end{proof}

Using the two equations $k=5n+8m-4$ and $3m=5n$ we get the following inequality.
\begin{equation} m \geq \frac{1}{11}k \label[ineq]{ineq:inap} \end{equation}

Combining the \Cref{ineq:inap} with the contrapositive of \Cref{lem:frac_sat} it
follows that, if less than $1-\delta$ fraction of the clauses of $\phi$ are
satisfied, then the number of lines needed is more than $k+\apconst\delta
k=(1+\apconst \delta)k$. Thus our reduction ensures both of the properties we
claimed.

\subsubsection{Proof of \Cref{th:cuhs}}
\label{App:cuhs}

\thcuhs*
\begin{proof}
Assume that there is such an algorithm $\mathbb{A}$. Then we show that, given an
instance $(P,k)$ of UHS, it is possible to decide in polynomial time whether it
can be hit by at most $k$ lines, or any hitting set corresponding to it needs
more than $(1+\con)k$ lines. Thus from \Cref{th:inapproxHS} it follows that
\PNP.

Given the instance $(P,k)$ of UHS, we call $\mathbb{A}$ on $P,r,c$ for each pair
of positive integers $r$ and $c$ such that $r+c=k$. If $\mathbb{A}$ returns yes
for at least one pair, then we return ``yes'' for the instance of UHS.
Otherwise, we return ``no'' for the instance of UHS. 

If the instance $(P,k)$ admits a hitting set of size at most $k$, fix such a
hitting set $L$. Let $L$ contains $r'$ horizontal and $c'$ vertical lines. Then
for $r=r'$ and $c=k-c'\geq r'$, $\mathbb{A}$ returns ``yes'', and so we return
``yes'' for the instance of UHS. If the instance $(P,k)$ does not admit a
hitting set of size at most $(1+\con)k$, then there is no hitting set for
$\Seg(P)$ with at most $(1+\con)r$ horizontal and $(1+\con)c$ vertical lines
such that $r+c=k$. Thus $\mathbb{A}$ returns ``no'' on all $(r,c)$ pairs on
which it is invoked, and so we return ``no'' for the instance of UHS.
\end{proof}

\subsubsection{Proof of \Cref{th:inap_CLADS}}\label{App:CLADS}

\thinapCLADS*
\begin{proof}
Assume that such an algorithm $\mathbb{B}$ does exist. Then we show that, given
an instance $(P,r,c)$ of CUHS, it is possible to decide in polynomial time
whether there is a set with at most $c$ vertical and $r$ horizontal lines that
hits the segments in $\Seg(P)$, or there is no hitting set for $\Seg(P)$ using
at most $(1+\con)c$ vertical and $(1+\con)r$ horizontal lines. By \Cref{th:cuhs}
this is true only if \PNP and hence the result follows.

Given the instance $(P,r,c)$ of CUHS, if $r=0$ or $c=0$, CUHS can be solved in
polynomial time. Otherwise, we construct a set of squares $S'$ by taking a unit
square centered at $p$ for each point $p\in P$. We invoke $\mathbb{B}$ on $S=S',
h=r+1, w=c+1$ and $\varepsilon=\con'$. If $\mathbb{B}$ returns ``yes'', we
return ``yes''. Otherwise, we return ``no''.

Suppose the instance $(P,r,c)$ does admit a hitting set with at most $c$
vertical and $r$ horizontal lines. Then by \Cref{cl:forward} there is an output
layout $\layp$ with $H(\layp) \leq r + 1 + \con'$ and width $W(\layp) \leq c + 1
+ \con'$. Thus $\mathbb{B}$ returns ``yes'', and so we return ``yes'' for the
instance of CUHS.

Suppose the instance $(P,r,c)$ does not admit a hitting set with at most
$(1+\con)c$ vertical and $(1+\con)r$ horizontal lines. Then by
\Cref{cl:reverse}, for any output layout $\layp$ of the squares in $S$, either
$H(\layp) > (1+\con)r + 1$ or $W(\layp) > (1+\con)c + 1$. Now as $c\geq 1$,
$(1+\con)c + 1 = c+1+4\con'c = c+\con'c+1+2\con'c+\con'c >
c+\con'c+1+2\con'+{\con'}^2=(1+\con')(c+1+\con')$. Similarly, as $r\geq 1$,
$(1+\con)r + 1 > (1+\con')(r+1+\con')$. Hence either $H(\layp) >
(1+\con')(r+1+\con')$ or $W(\layp) > (1+\con')(c+1+\con')$. Thus at least one of
the two conditions $W(\layp) \leq (1+\con')(w + \varepsilon)$ and $H(\layp) \leq
(1+\con')(h+\varepsilon)$ is false, and $\mathbb{B}$ returns ``no'', and so we
return ``no''.
\end{proof}

\end{document}